\documentclass[12pt]{article}
\usepackage[utf8]{inputenc}
\usepackage{amsmath,amsthm,amssymb}
\usepackage{graphicx}
\usepackage{mdframed}
\usepackage[table,xcdraw]{xcolor}
\usepackage{booktabs}
\usepackage{caption}
\usepackage{mathtools}
\usepackage[margin=1in]{geometry}
\usepackage[title]{appendix}
\usepackage{natbib}
\usepackage{hyperref}
\usepackage{pdfpages}
\usepackage{subfig}
\usepackage{accents}
\usepackage{authblk}
\usepackage{xr}
\usepackage{algorithm}
\usepackage{algpseudocode}
\usepackage{tikz}
\usetikzlibrary{arrows.meta}
\usetikzlibrary{quotes}
\hypersetup{
    colorlinks=true,
    citecolor=blue,
    linkcolor=blue,
    filecolor=blue,      
    urlcolor=magenta
    }
\allowdisplaybreaks
\def\bs{\boldsymbol}
\def\lb{\left\{ }
\def\rb{\right\}}
\def\tx{\textrm}

\def\E{\mathbb E}
\def\P{\mathbb P}
\def\Q{\mathbb Q}

\def\Udot{U(\cdot,\cdot)}
\def\Vdot{V(\cdot,\cdot)}
\def\Ujdot{U^J(\cdot,\cdot)}
\def\Vjdot{V^J(\cdot,\cdot)}
\def\EP{\E^P}
\def\KL{\tx{KL}}

\def\iid{\overset{iid}{\sim}}
\def\EQ{\E^Q}
\def\Ehat{\widehat E}
\def\diff{ \tx{d}}

\def\N{\mathcal N}
\def\TV{\tx{TV}}
\def\genq{\underset{\sim}{>}}

\newcommand{\ind}{\perp\!\!\!\!\perp} 
\theoremstyle{definition}

\newtheorem{theorem}{Theorem}

\newtheorem{prop}{Proposition} 

\newtheorem{corr}{Corollary}

\title{Besag-Clifford e-values for unnormalized testing}
\author[1]{Alexander Dombowsky}
\author[1,2]{Barbara E. Engelhardt}
\author[3,4]{Aaditya Ramdas}
\affil[1]{Gladstone Institute of Data Science and Biotechnology}
\affil[2]{Department of Biomedical Data Science, Stanford University}
\affil[3]{Department of Statistics and Data Science, Carnegie Mellon University}
\affil[4]{Machine Learning Department, Carnegie Mellon University}

\date{\today}

\begin{document}

\maketitle

\begin{abstract}
    Unnormalized probability distributions are frequently used in machine learning for modeling complex data generating processes. Though Markov chain Monte Carlo (MCMC) algorithms can approximately sample from unnormalized distributions, intractability of their normalizing constants renders likelihood ratio testing infeasible. We propose to use the parallel method of Besag and Clifford to generate samples that are exchangeable with the data under the null, to then generate valid e-values for any number of iterations or algorithmic steps. We show that as the number of samples grows, these Besag-Clifford e-values constructed using the unnormalized likelihood ratio are actually log-optimal up to a multiplicative term that diminishes with the mixing time of the Markov chain. Additionally, averaging over the output of multiple chains retains validity while increasing the e-power. We extend Besag-Clifford e-values to the general problem of unnormalized test statistics, which allows application to composite hypotheses, uncertainty quantification, generative model evaluation, and sequential testing. Through simulations and an application to galaxy velocity modeling, we empirically verify our theory, explore the impact of autocorrelation and mixing, and evaluate the performance of Besag-Clifford e-values.
\end{abstract}

\newpage

\section{Introduction}

Let $X = (X_1, \dots, X_n)$ be a random vector of independent and identically distributed (iid) observations. For a simple null hypothesis $\mathcal P = \{ P \}$ and alternative hypothesis $\mathcal Q = \{ Q \}$, the Neyman-Pearson lemma states that the uniformly most powerful level-$\alpha$ test is obtained by thresholding the \textit{likelihood ratio} 
\begin{equation} \label{eq:likelihood-ratio-definition}
    E(X) = \frac{q(X)}{p(X)},
\end{equation}
where $p(x)$ and $q(x)$ are the probability densities or mass functions (pdfs or pmfs) of $P$ and $Q$, respectively. Many probability distributions used in practice are only known up to a normalization constant, including products of experts (PoEs) \citep{hinton1999products}, Markov random fields (MRFs) \citep{kindermann1980markov}, and  restricted Boltzmann machines (RBMs) \citep{fischer2012introduction}. Bayesian models fall into this category as well, since the marginal likelihood, or the integral of the likelihood over the prior distribution, is generally not available in a closed form. In the case of an \textit{unnormalized null distribution}, i.e., $p(x) \propto g(x)$ with $K:=\int g(x) \diff x$ being unknown, one cannot calculate $E(X)$.

$E(X)$ is the simplest example of an e-variable \citep{ramdas2025hypothesis}. Given a null hypothesis specified by a set of distributions $\mathcal P$, an e-variable for $\mathcal P$ is a nonnegative test statistic whose expected value under every $P \in \mathcal P$ is no more than $1$. The likelihood ratio satisfies this criterion, as $\EP[E(X)] = \int \{q(x)/p(x)\}p(x) \diff x = \int q(x) \diff x = 1$. The realization of an e-variable after observing data is known as an \textit{e-value}, in this case, $E(x)$. By Markov's inequality, for any e-variable $T(X)$, $\textbf{1}(T(X) \geq 1/\alpha)$ is a level-$\alpha$ test: $\E^P[\textbf{1}(T(X) \geq 1/\alpha)] = P(T(X) \geq 1/\alpha) \leq \E^P[T(X)]/(1/\alpha) \leq \alpha$. In addition, we can combine e-values through averaging or, if they are independent under the null, taking their product. E-values are appealing for post-hoc and sequential testing; in these settings, approaches based on e-values always guarantee type-1 error control, whereas p-values based methods may not \citep{wang2022false,koning2023post,xu2024post,ramdas2025hypothesis}.

The likelihood ratio is not the only e-variable that exists for simple hypotheses. One could create another e-variable by replacing $q(X)$ with $w(X)$ in equation \eqref{eq:likelihood-ratio-definition}, where $w(x)$ is another pdf or pmf. Since we reject $\mathcal P$ when $T(x) \geq 1/\alpha$, it is appealing for an e-variable to be large when the null is false. Formally, one tries to design e-variables with large \textit{e-power}, defined as $\E^Q[\log T(X)]$ for any e-variable $T(X)$ and alternative $Q$. Importantly, $E(x)$ is the unique \textit{log-optimal} e-value for testing simple hypotheses, i.e., for any other e-value $E^\prime(x)$, $\EQ[\log(E^\prime(X)/E(X))] \leq 0$, with e-power $\EQ[\log E(X)] = \tx{KL}(Q, P)$, the Kullback-Liebler (KL) divergence between $Q$ and $P$. \cite{larsson2025numeraire} showed that a log-optimal e-variable called the numeraire exists for any composite null hypothesis $\mathcal P$, and any alternative $Q$. In the absence of an optimal Neyman-Pearson test, which is typical for general composite nulls,
we can obtain a level-$\alpha$ test by thresholding the numeraire above $1/\alpha$.

Complications arise with unnormalized distributions.
An unnormalized likelihood ratio (ULR) $T(x) = q(x)/g(x)$ is generally not an e-value, since $ \E^P[T(X)] = 1/K$, so we have no guarantee that $\textbf{1}(T(X) \geq 1/\alpha)$ is a level-$\alpha$ test. However, large values of $T(x)$ can be interpreted as evidence against the null. The ULR is a specific example of \textit{unnormalized testing}: rejecting the null when a nonnegative function $T(x)$ is large, but it may be the case that $\E^P[T(X)] > 1 $ or $\E^P[T(X)]$ is unknown. 

A natural progression is to consider the case that $q(x)$ is also known only up to its normalizing constant. Thus, $q(x) \propto h(x)$, so $q(x) = h(x)/L$, where $L = \int h(x) \diff x< \infty$. The null expectation of the test statistic $T(X) = h(X)/g(X)$ is $\EP[T(X)] = L/K \neq 1$. 

Unnormalized test functions also appear in composite hypothesis testing. For a simple null and composite alternative, if $\hat q$ is a pre-trained model for the alternative hypothesis, then $T^\prime(x) = \hat q(x)/p(x)$ is an e-value. And yet, $T(x) = \hat q^{\tx{MLE}}(x)/p(x)$, where $q^{\tx{MLE}}(x)$ is the maximized likelihood over the alternative, is conceptually appealing for testing and dominates $T^\prime(x)$, but $T(x)$ is often not an e-value and is thus technically unnormalized.

In this article, we use simulations from the null to normalize $T(x)$ into an e-value, allowing for valid unnormalized testing. 
However, we do not assume that we can simply sample exactly from $P$, which would be a strong assumption. Instead, we will assume that we can construct a Markov Chain with $P$ as its mixing distribution, but without any knowledge or bound on the mixing times, which is a much weaker assumption. For example, one can construct a Markov chain whose unique stationary distribution is $P$ with the Metropolis-Hastings algorithm \citep{metropolis1953equation}, so long as the proposal mechanism is chosen carefully.

Our approach is motivated by the following property. If $Y^{(1)}, \dots, Y^{(M)}$ are generated in a manner so that, if $X \sim P$, then $(X, Y^{(1)}, \dots, Y^{(M)})$ are exchangeable, then
\begin{equation} \label{eq:gof-p-value}
    \hat P^{\tx{gof}} = \frac{1 + \sum_{m=1}^M \textbf{1}(T(Y^{(m)}) \geq T(X))}{M+1}
\end{equation}
is a valid p-value for any nonnegative test statistic $T(X)$. $\textbf{1}(\hat P^{\tx{gof}} \leq \alpha)$ is a level-$\alpha$ test and known as a goodness-of-fit test or a Monte Carlo test. The choice of $T(x)$ is arbitrary but has an impact on power; it is often advantageous to include information on the alternative, $Q$, when specifying $T(x)$. The numerator of equation \eqref{eq:gof-p-value} computes the rank of $T(X)$ among $(T(X), T(Y^{(1)}), \dots, T(Y^{(M)}))$, and so we reject $P$ when $T(X)$ is large relative to $T(Y^{(m)})$ for $m=1, \dots, M$. There are several methods for obtaining MCMC samples that are exchangeable with $X$ under the null, most notably the parallel method of \cite{besag1989generalized}, which we review below (Section~\ref{section:MCMC}). Since $\hat P^{\tx{gof}}$ is a p-value, there is no general guarantee for post-hoc or data-dependent stoppage validity. However, the ratio of $T(X)$ to the mean of $T(X), T(Y^{(1)}), \dots, T(Y^{(M)})$, or the ``soft-rank," is an e-variable \citep{wang2022false} under the null exchangeability assumption underlying equation \eqref{eq:gof-p-value}.
\begin{prop} \label{prop:MC-test}
    For testing a null $\mathcal P$ against an alternative $\mathcal Q$ using observed data $X$, if $(X, Y^{(1)}, \dots, Y^{(M)})$ are exchangeable for any $P \in \mathcal P$, then
    \begin{equation} \label{eq:MC-test}
   \Ehat_M(X) = \frac{(M+1)T(X)}{T(X) + \sum_{m=1}^M T(Y^{(m)})}
    \end{equation}
    is an e-variable for any $M \in \mathbb N$.
\end{prop} 
\begin{proof}
    This is a natural consequence of exchangeability. With the convention that $0/0=0$,
    \begin{gather*}
        \frac{T(X) + T(Y^{(1)}) + \dots + T(Y^{(M)})}{T(X) + \sum_{m=1}^M T(Y^{(m)})} \leq 1 \tx{ a.s. } P 
        \\
        \implies (M + 1) \EP \bigg [ \frac{T(X)}{T(X) + \sum_{m=1}^M T(Y^{(m)})} \bigg ] \leq 1,
    \end{gather*}
    where the inequality is an equality if $P(T(X)>0) = 1$. 
\end{proof}

We refer to $\Ehat_M(x)$ as a \textit{Besag-Clifford e-value}.
There are two main considerations that go into designing $\Ehat_M(X)$: (a) the test statistic $T(x)$ and (b) the sampling algorithm. This article studies these two considerations in detail. In particular, for testing simple hypotheses, we propose setting $T(x) \propto E(x)$, e.g., $T(x) = q(x)/g(x)$, the latter quantity being the ULR. Under this choice of $T(x)$, the denominator of $\Ehat_M(X)$ functions as an estimator of $\EP[q(X)/g(X)] = 1/K$. As $M \to \infty$, we show that $\Ehat_M(X)$ converges in probability to a random scalar multiple of the true likelihood ratio $E(X)$.
To motivate our results, we first discuss a simpler case: suppose that we can sample $Y^{(1)}, \dots, Y^{(M)} \iid P$ independently of $X$, i.e., our MCMC algorithm samples exactly and independently from the target distribution. Then if $T(x) \propto E(x)$, the associated Besag-Clifford e-variable converges almost-surely to the log-optimal e-variable.

\begin{theorem} \label{thm:ergodic-log-optimal-iid}
    Suppose that $T(x) < \infty$ and $T(x) \propto q(x)/p(x)$ for all $x \in \mathbb R^n$. Let $X \sim W$, where $W \in \{ P, Q\}$. Then,
    \begin{equation*}
        \Ehat_M(X) \overset{a.s.}{\longrightarrow} E(X).
    \end{equation*}
\end{theorem}

The proof is detailed in Section~\ref{supp:proofs}, but the result essentially follows from the strong law of large numbers. 
We emphasize that validity holds \textit{for all} $M$ by Proposition~\ref{prop:MC-test}, but Theorem~\ref{thm:ergodic-log-optimal-iid} states that setting $T(x) \propto q(x)/p(x)$ results in the highest e-power as $M \to \infty$, analogous to the e-power motivation for choosing the likelihood ratio as an e-variable over other choices.

As an example, suppose we test $\mathcal P = \{ \tx{Poisson}(1)\}$ against $\mathcal Q = \{ \tx{Poisson}(1.1)\}$. We have that $ p(x) \propto \{1/\prod_{i=1}^n /x_i!\}$ and $q(x) \propto 1.1^{\sum_{i=1}^n x_i}/\{\prod_{i=1}^n x_i!\}$, so we set $T(x) = 1.1^{\sum_{i=1}^n x_i}$. Across $1,000$ independent experiments, we simulate $n=100$ observations from a $\tx{Poisson}(1.1)$ distribution, and we compute $\Ehat_M(X)$ for $M \in \{10,100,500,1000 \}$, as well as the true likelihood ratio $E(X)$. Comparing $\log E(X)$ with $\log \Ehat_M(X)$ across all choices of $M$, as $M$ increases, the values of $\Ehat_M(X)$ resemble the true likelihood ratio (Figure~\ref{fig:Poisson-simulations}), agreeing with Theorem~\ref{thm:ergodic-log-optimal-iid}.

\begin{figure}[t]
    \centering
    \includegraphics[scale=0.25]{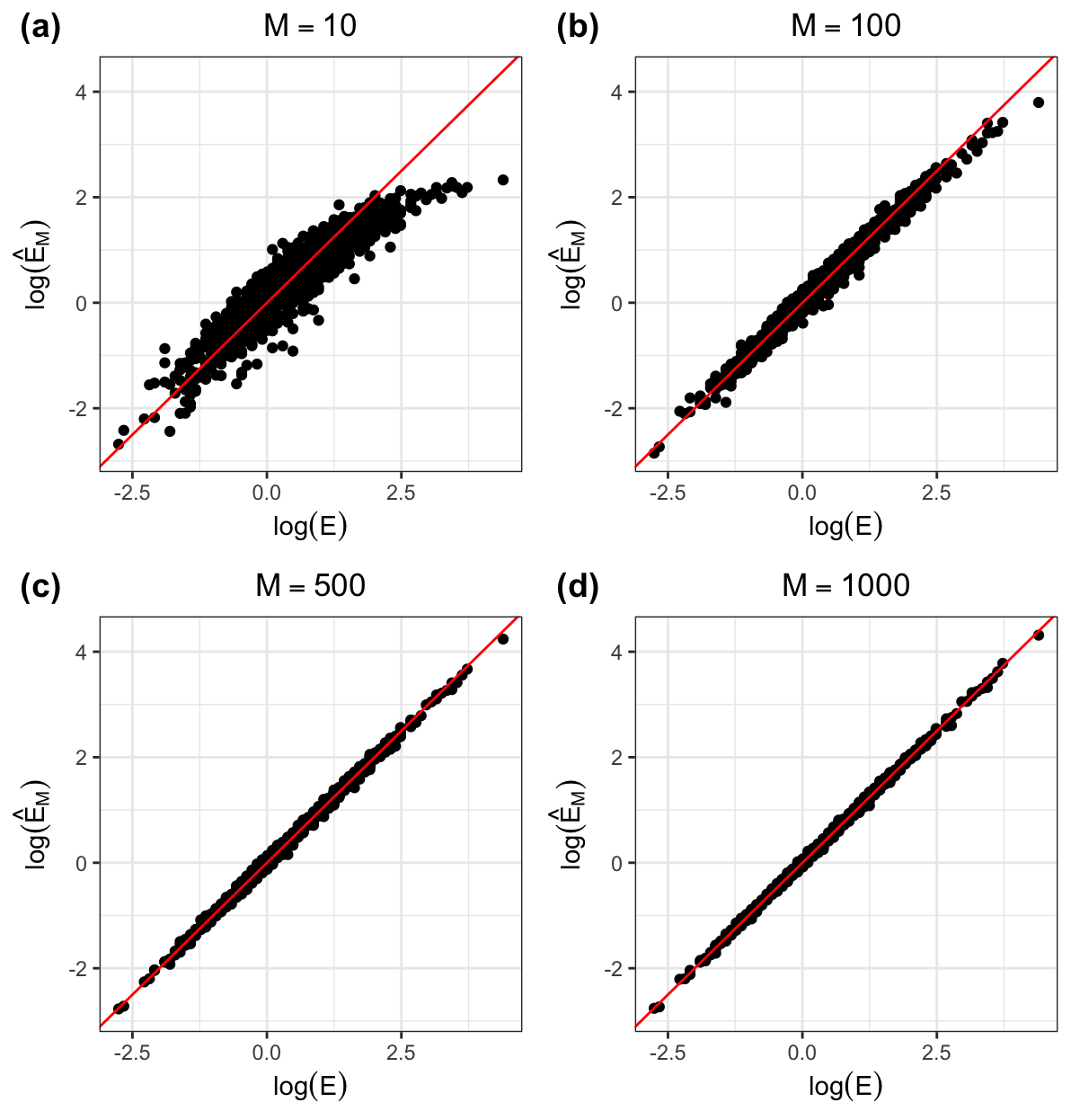}
    \caption{Scatter plots of $1,000$ simulated values of the Poisson$(1)$ and Poisson$(1.1)$ likelihood ratio and $\Ehat_M$ for $M \in \{ 10, 100, 500, 1000 \}$, where $n=100$, and the data are simulated from the alternative. The red line in each plot indicates the identity map.}
    \label{fig:Poisson-simulations}
\end{figure}

When iid samples from $P$ cannot be obtained, but we can sample $Y^{(1)}, \dots, Y^{(M)} \iid U$, we can use a similar proof technique to Theorem~\ref{thm:ergodic-log-optimal-iid} to show that
\begin{equation} \label{eq:biased-theorem-1}
    \Ehat_M(X) \overset{a.s.}{\longrightarrow} \frac{E(X)}{\E^U[E(Y)]},
 \end{equation}
 so long as the denominator, defined as $\E^U[E(Y)] = \int \{ q(y)/p(y) \} u(y) \diff y$, is nonzero. Since $U \neq P$, $\Ehat_M(X)$ converges to a scaled version of the log-optimal e-variable $E(X)$, and the departure from log-optimality is entirely determined by the relationship of $U$ and $P$. If $U \approx P$, then $w(y)/p(y) \approx 1$, and $\E^U[E(Y)] \approx 1$. We focus on a similar setting in this paper, as we consider distributions $P$ that cannot be directly sampled with standard Monte Carlo methods. 

The remainder of this work focuses on the construction and theoretical properties of our \textit{Besag-Clifford e-values}. Section~\ref{section:MCMC} reviews the main concepts of MCMC and the parallel method of \cite{besag1989generalized}, Section~\ref{section:related-work} discusses related literature, and Section~\ref{section:contributions} details our primary contributions. Unnormalized likelihood ratios using Besag-Clifford e-values are proposed in Section~\ref{section:methods}, along with a method to boost e-power using multiple chains and results specific to autoregressive processes. Sections~\ref{section:composite} and~\ref{section:sequential} show how to incorporate composite hypotheses and sequential experiments into our framework. Sections~\ref{section:examples} and~\ref{section:simulations} focus on examples and the performance of Besag-Clifford e-values on simulated data. An application of Besag-Clifford e-values to galaxy velocity data contained in Section~\ref{section:application}. Lastly, concluding remarks and extensions are given in Section~\ref{section:discussion}. 

\subsection{Markov chain Monte Carlo} \label{section:MCMC}

A Markov chain is a discrete stochastic process $Y^{(0)}, Y^{(1)}, Y^{(2)}, \dots,$ that satisfies the Markov property: 
\begin{equation*}
    Y^{(j)} \ind (Y^{(0)}, \dots, Y^{(j-2)}) \mid Y^{(j-1)}.
\end{equation*}
The forward transition kernel $U(y,y^\prime)$ is the conditional distribution of a new state $Y^\prime$ given an initial state $Y$, which admits a density $u(y,y^\prime)$. Forward simulation of a Markov chain occurs via $Y^{(0)} \sim \Pi$ and $Y^{(j)} \mid Y^{(j-1)} \sim U(Y^{(j-1)},\cdot)$. The $J$-step transition kernel is the conditional distribution of $Y^{(J)} \mid Y^{(0)}$, and is denoted $U^J(y,y^\prime)$, with $u^J(y,y^\prime)$ being its density. A distribution $P$ is a \textit{stationary distribution} of a Markov chain if sampling $Y^{(0)} \sim P$ and $Y^{(1)} \mid Y^{(0)} \sim U(Y^{(0)},\cdot)$ implies $Y^{(1)} \sim P$, and so sequential sampling generates identically distributed random vectors from $P$. 

The reverse transition kernel $V(y^\prime,y)$ is a probability distribution that satisfies the \textit{detailed balance equations}
\begin{equation} \label{eq:detailed-balance}
    p(y)u(y,y^\prime) = p(y^\prime)v(y^\prime,y) \tx{ for all } y,y^\prime \in \mathbb R^n,
\end{equation}
where $v(y^\prime,y)$ is the density of $V(y^\prime, y)$. Stationarity of $P$ is preserved under simulation of the reverse transition kernel. If $V(y^\prime, y) = U(y^\prime, y)$, then the Markov chain is \textit{reversible}. If a stationary Markov chain is \textit{irreducible} and \textit{aperiodic}, sample averages of functions of $Y^{(0)}, \dots, Y^{(M)}$ converge in probability to the true expectation under $P$; these concepts are not critical for understanding our method, but we refer to reader to \cite{levin2017markov} and references therein for their formal definitions.

Therefore, given a sample $Y^{(0)} \sim P$ and a stationary Markov chain for $P$, one can generate autocorrelated samples via iterative simulation from the forward (or backward) transition kernels. In the context of testing, if we initialize a Markov chain at $Y^{(0)} = X$, then under the null hypothesis, the iterates $Y^{(1)}, \dots, Y^{(M)} \sim P$. However, validity of $\hat P^{\tx{gof}}$ relies on exchangeability, not stationarity, and $Y^{(1)}, \dots, Y^{(M)}$ are not exchangeable despite having the same marginal distribution. 

The parallel method of \cite{besag1989generalized} uses stationarity and a backward-forward sampling scheme to ensure exchangeability of $(X,Y^{(1)}, \dots, Y^{(M)})$, and proceeds as follows. First, we evolve the chain backwards in time using $V(X,\cdot)$ for $J$ steps, resulting in $Y^{(0)}$. We then run the chain forwards independently for $J$ steps using $U(Y^{(0)},\cdot)$, yielding $Y^{(1)}, \dots, Y^{(M)}$. Note that these samples can be generated in parallel after $Y^{(0)}$ is obtained. Exchangeability follows by noting that $(X,Y^{(1)}, \dots, Y^{(M)})$ are conditionally independent given $Y^{(0)}$, and that equation \eqref{eq:detailed-balance} implies that $X \mid Y^{(0)} \sim U^J(Y^{(0)},\cdot)$ if $X \sim P$. The parallel method is detailed in Algorithm~\ref{alg:besag-clifford-parallel}, and can be visualized with a directed graph (Figure~\ref{fig:Besag-Clifford-DAG}) \citep{barber2022testing}. When $X \not \sim P$, then $X,Y^{(1)}, \dots, Y^{(M)}$ are not excheangeable, giving our method power under the alternative.

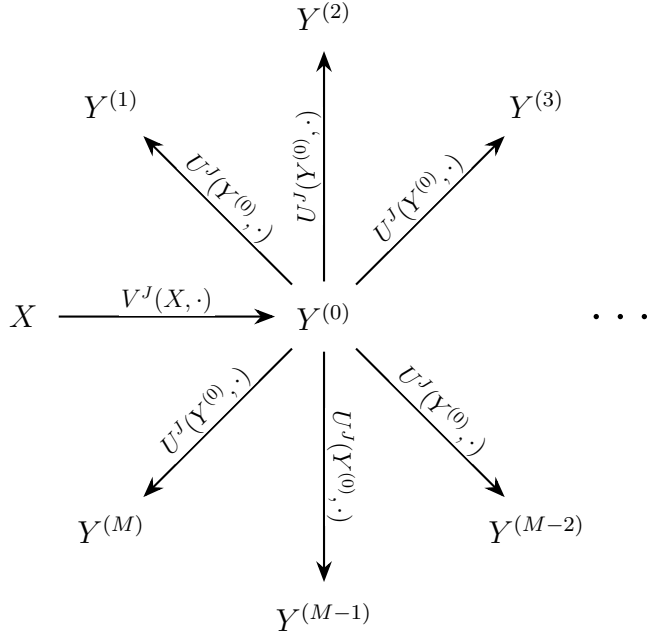
\begin{figure}[t]
\centering
\begin{tikzpicture}[
    arrow/.style={
        -{Stealth[scale=1.2]}, 
        thick, 
        shorten <=4pt, 
        shorten >=4pt
    },
    edge_label/.style={
        midway, 
        scale=0.8, 
        fill=white,
        inner sep=1pt
    }
]

    \node (Center) at (0,0) {$Y^{(0)}$};
    \node (X) at (180:4cm) {$X$};
    
    \node (Y1) at (135:4cm) {$Y^{(1)}$};
    \node (Y2) at (90:4cm)  {$Y^{(2)}$};
    \node (Y3) at (45:4cm)  {$Y^{(3)}$};
    
    \node at (0:4cm) {\LARGE $\dots$};

    \node (YM2) at (315:4cm) {$Y^{(M-2)}$};
    \node (YM1) at (270:4cm) {$Y^{(M-1)}$};
    \node (YM)  at (225:4cm) {$Y^{(M)}$};
    \draw[arrow] (X) -- (Center) 
        node[edge_label, above] {$V^J(X,\cdot)$};
    \foreach \target in {Y1, Y2, Y3, YM2, YM1, YM} {
        \draw[arrow] (Center) -- (\target) 
            node[edge_label, sloped, above] {$U^J(Y^{(0)},\cdot)$};
    }

\end{tikzpicture}
\caption{Graphical representation of Algorithm~\ref{alg:besag-clifford-parallel}. Starting from $X$, we evolve the chain backwards for $J$ steps to generate $Y^{(0)}$. From $Y^{(0)}$, the chain is evolved forwards for $J$ steps in parallel, producing $Y^{(1)}, \dots, Y^{(M)}$. $X$ and the draws are exchangeable when $X \sim P$.}
\label{fig:Besag-Clifford-DAG}
\end{figure}

\begin{algorithm}[t]
\caption{Parallel Method of \cite{besag1989generalized}}
\label{alg:besag-clifford-parallel}
\begin{algorithmic}[1]
\Require{$J$-step forward transition kernel $\Ujdot$, $J$-step backward transition kernel $\Vjdot$, initial state $X$.}
    \State Sample $Y^{(0)} \sim V^J(X, \cdot)$.
    \For{$m=1, \dots, M$} 
        \State Sample $Y^{(m)} \sim U^J(Y^{(0)},\cdot)$;
    \EndFor
\State \Return $(Y^{(1)}, \dots, Y^{(M)})$.
\end{algorithmic}
\end{algorithm}

If the chain is reversible, then forward and backward simulation involve the same conditional distribution. When the normalizing constant of $p$ is unknown, we can apply the parallel method with an MCMC algorithm. Fundamental MCMC algorithms include the Metropolis-Hastings algorithm, Gibbs sampling \citep{geman1984stochastic}, the Metropolis-adjusted Langevin algorithm (MALA) \citep{besag1994comment}, and Hamiltonian Monte Carlo (HMC) \citep{duane1987hybrid,neal2011mcmc}; each of these algorithms only requires the unnormalized kernel $g(x)$ for simulation. For distribution families that admit sufficient statistics, exact co-sufficient sampling \citep{bartlett1937properties, engen1997stochastic} and approximate co-sufficient sampling \citep{barber2022testing,zhu2023approximate, bhaduri2025conditioning} can be used to generate samples that are exchangeable under the null for composite hypotheses.

\subsection{Related work} \label{section:related-work}

Our approach leverages the algorithms for Monte Carlo testing \citep{besag1989generalized}, but the validity of $\hat P^{\tx{gof}}$ and $\Ehat_M(X)$ under an exchangeable joint distribution follows from results on permutation tests \citep{fisher1956mathematics}. Large MCMC sample properties for p-value based Monte Carlo significance tests were established in \cite{howes2026markov}, including convergence of $\hat P^{\tx{gof}}$ as $M \to \infty$. A sequential variant of Monte Carlo testing with $\hat P^{\tx{gof}}$ was detailed in \cite{besag1991sequential}.
The vast majority of existing work on Monte Carlo significance testing uses p-values, rather than e-values. We build our framework around the latter due to the appealing properties for post-hoc and sequential testing.

More specifically, this paper contributes to the growing literature on e-values and e-processes \citep{shafer2001probability,vovkevalues2021,grunwald2024safe,howard2020time}; see \cite{ramdas2025hypothesis} for an overview.
$\Ehat_M(X)$ is an example of a soft-rank e-variable \citep{wang2022false,koning2023post,balinsky2024enhancing}, though our construction via MCMC simulation is novel. By proposing novel randomized e-values, we expand on the concepts of universal inference \citep{wasserman2020universal}, log-optimality \citep{larsson2025numeraire}, and testing by betting \citep{shafer2021testing, waudbysmith2023estimating}. 

With regard to unnormalized distributions, \cite{wei2022high} designs e-values for multiple testing problems in Markov random fields (MRFs) \citep{kindermann1980markov}, and  \cite{martinez2025sequential} uses the kernelized Stein discrepancy to define a class of e-values and e-processes; neither of these approaches uses null exchangeability algorithms. Several recent works have also drawn on the Besag-Clifford methodology. \cite{ramdas2023permutation} propose a permutation test with an analogous probabilistic structure to Algorithm~\ref{alg:besag-clifford-parallel}, and  \cite{fischer2025sequential} introduce sequential Monte Carlo testing by betting. In the latter work, the authors assume that (a) one can generate $T(Y^{(1)}), \dots, T(Y^{(M)})$ so that they are exchangeable conditional on $T(X)$, and (b) the null hypothesis is that $T(X),T(Y^{(1)}),\dots, T(Y^{(M)})$ are exchangeable. Indeed, our set-up is a special case of these conditions. We also focus on a similar alternative to theirs when specifying a Besag-Clifford e-value in equation \eqref{eq:MC-test}, namely that $T(X)$ is stochastically larger than $T(Y^{(m)})$. However, the goals and methods of each approach differ, their construction is an e-process with respect to indicators $\textbf{1}(T(Y^{(m)}) \geq T(X))$, whereas the Besag-Clifford e-process we introduce is an e-process with respect to new data and draws.

Simulations from unnormalized models have previously been used for maximum likelihood estimation, most notably in \cite{geyer1992constrained} and \cite{geyer1994convergence} with asymptotic guarantees. We refrain from using this methodology in the interest of validity of $\Ehat_M(X)$ for any $M \in \mathbb N$, as established in Proposition~\ref{prop:MC-test}. Sampling from unnormalized distributions is a keystone of modern Bayesian inference \citep{hoff2009first}, as the posterior is proportional to the product of the likelihood function and the prior distribution, both of which are known. Averages of MCMC samples, such as that in the denominator of $\Ehat_M(X)$, are routinely calculated for parameter estimation in Bayesian statistics. Parameters in unnormalized models may be estimated through score-matching \citep{hyvarinen2005estimation}, contrastive divergence \citep{hinton2002training}, and noise-contrastive estimation \citep{gutmann2012estimation}. We use the former in for a test for between competing pre-trained PoE models (Section~\ref{section:application}).

\subsection{Summary of contributions} \label{section:contributions}

Below, we summarize the main properties and extensions of Besag-Clifford e-values. 
\begin{enumerate}
    \item For a simple null hypothesis $\mathcal P = \{P \}$ and a simple alternative hypothesis $\mathcal Q = \{ Q \}$, suppose we sample $(Y^{(1)}, \dots, Y^{(M)})$ using Algorithm~\ref{alg:besag-clifford-parallel} and a general MCMC algorithm. As $M \to \infty$, a Besag-Clifford e-value with $T(x) \propto E(x)$ is log-optimal up to a random scalar $\Delta^J(Y^{(0)})$ that depends on $Y^{(0)}$ and $J$, in the sense that the probability-limit of $\Delta^J(Y^{(0)}) \Ehat_M(X)$ is the likelihood ratio $E(X)$.
    \item The scalar $\Delta^J(Y^{(0)})$ is the likelihood ratio of $q^*(Y^{(0)}) = \int q(y)V^J(y,Y^{(0)}) \diff y$ and $p(Y^{(0)})$. In some cases, we can write $Y^{(0)}$ as a random function of $X$. We can then express the scalar in terms of $X$, denoted $\Delta^J(X)$. $\Delta^J(X)$ is an exact e-variable for $P$ and $1/\Delta^J(X)$ is an exact e-variable for $Q$, i.e.,
    \begin{equation*}
        \EP[\Delta^J(X)] = \EQ \bigg [ \frac{1}{\Delta^J(X)} \bigg] = 1;
    \end{equation*}
    thus, we interpret $\Delta^J(x)$ as a randomized e-value. The e-power of $\Delta^J(X)$ is related to the mixing of the MCMC algorithm, and is arbitrarily small when $J$ is set to be the mixing time of the chain with respect to the KL divergence.
    \item For testing the mean or variance of a Gaussian distribution, we derive closed form expressions of $\Delta^J(X)$ when the autoregressive process of order $1$ \citep{brockwell1991time} is used in Algorithm~\ref{alg:besag-clifford-parallel}. Furthermore, the e-power of $\Delta^J(X)$ scales as $\mathcal O(\phi^{2J})$ in both testing problems, where $-1 < \phi < 1$ is the $1$-step autocorrelation.  
    \item For composite nulls $\mathcal P$ with finite cardinality, we show that generating samples from each $P \in \mathcal P$ using Algorithm~\ref{alg:besag-clifford-parallel} can be used to construct a composite Besag-Clifford e-value and a confidence region.
    \item We use these results to design a general class of \textit{Besag-Clifford e-processes} for testing simple (or finite cardinality) hypotheses, ensuring control of the type-1 error probability at any data-dependent stopping time. 
    \item If $T(X)$ is an e-variable but there exists a hypothesis for which it is under-powered ($\EQ[T(X)] \leq 1$), then $\EQ[\log \Ehat_M(X)] \geq \EQ[\log T(X)]$ for any $M$ when $Y^{(m)} \iid P$.
\end{enumerate}

\section{Intractable likelihood ratio testing} \label{section:methods}

\subsection{Simple hypotheses}

Assume that we are testing a simple null hypothesis, $\mathcal P = \{ P \}$, against a simple alternative hypothesis, $\mathcal Q = \{ Q \}$, where $Q \ll P$. Recall that $E(X) = q(X)/p(X)$ is the log-optimal e-variable for this test. We use this section to discuss the asymptotic behavior of $\Ehat_M(X)$ when applying the parallel method (Algorithm~\ref{alg:besag-clifford-parallel}) to generate exchangeable samples. First, for a fixed $J$ number of steps, the probability-limit of $\Ehat_M(X)$ with $T(x) \propto q(x)/p(x)$ is a scalar multiple of $E(X)$ as $M \to \infty$.

\begin{theorem} \label{thm:Besag-parallel}
    Suppose that $(Y^{(1)}, \dots, Y^{(M)})$ are simulated from the parallel method and $T(x) \propto q(x)/p(x)$. For any $y^\prime \in \mathbb R^n$, let 
    \begin{equation} \label{eq:parallel-remainder}
        \Delta^J(y^\prime) =  \int \frac{q(y)}{p(y)} u^J(y^\prime, y) \diff y, 
    \end{equation}
    and assume there exists a $J \in \mathbb N$ so that $0 < \Delta^J(y^\prime) < \infty$ for all $y^\prime \in \mathbb R^n.$ For $X \sim W \in \{P,Q\}$ and $Y^{(0)} \mid X \sim V^J(X, \cdot)$,
    \begin{equation} \label{eq:parallel-pr-limit}
        \Ehat_M(X) \overset{pr.} {\longrightarrow} E^J(X) = \frac{E(X)}{\Delta^J(Y^{(0)})}.
    \end{equation}
\end{theorem}

The proof of Theorem~\ref{thm:Besag-parallel} largely follows from the weak law of large numbers, since $Y^{(1)}, \dots, Y^{(M)}$ are iid conditional on $Y^{(0)}$. The quantity $\Delta^J(Y^{(0)})$ is the expectation of the likelihood ratio under the $J$-step forward transition kernel, initialized at $Y^{(0)}$. We interpret this quantity as a multiplicative bias explained by the difference between $\Udot$ and $P$, similar to the denominator of equation \eqref{eq:biased-theorem-1}. If $u(y,y^\prime) = p(y)$, for instance, then $\Delta^J(y^{(0)}) = 1$ and $\Ehat_M(X)$ converges to the likelihood ratio, as shown in Theorem~\ref{thm:ergodic-log-optimal-iid}. Equation \eqref{eq:parallel-pr-limit} is an almost-sure limit as well provided that $\Delta(y^{(0)})$ is continuous for all $y^{(0)}$, using a similar proof technique to Theorem~\ref{thm:ergodic-log-optimal-iid}.

Our main assumption in Theorem~\ref{thm:Besag-parallel} is that $\Delta^J(y^\prime) \in ( 0, \infty)$ for some $J \in \mathbb N$, which allows us to invert the Monte Carlo average. This condition is light in practice, since the detailed balance equations imply that
\begin{equation} \label{eq:delta-rewritten}
   \Delta^J(Y^{(0)}) = \frac{\int q(y)v^J(y,Y^{(0)}) \diff y}{p(Y^{(0)})}. 
\end{equation}
Like $E(X)$, $\Delta^J(Y^{(0)})$ is a likelihood ratio, but with alternative $q^*(y^\prime) := \int q(y)v^J(y,y^\prime) \diff y^\prime$, or the marginal distribution of $Y^{(0)}$ when $X \sim Q$ in Algorithm~\ref{alg:besag-clifford-parallel}. In summary, we only need absolute continuity of $Q^*$ with respect to $P$ for our assumption to hold. 

For some algorithms, e.g., autoregressive processes (Section~\ref{section:ar1}), we can express $Y^{(0)}$ as a random function of $X$, such as $Y^{(0)} = \phi X + \epsilon^{(0)}$, with $\epsilon^{(0)}$ being random additive noise. Plugging this into equation \eqref{eq:delta-rewritten} gives
\begin{equation*}
    \Delta^J(X) := \frac{q^*(\phi X + \epsilon^{(0)})}{p(\phi X + \epsilon^{(0)})}.
\end{equation*}
It follows that $ \EP[\Delta^J(X)] = \EQ [ 1/{\Delta^J(X)}] = 1$,
or that $\Delta^J(x)$ is an e-value for $P$ and $1/\Delta^J(x)$ is an e-value for $Q$. More generally, we have that $\E^{V^J}[\Delta^J(Y^{(0)}) \mid x]$ is an exact e-value for $P$ and $\E^{V^J}[1/\Delta^J(Y^{(0)}) \mid x]$ is an exact e-value for $Q$, where the expectation defining these two objects is with respect to $V^J(x,\cdot)$. 

We can also interpret $\Delta^J(Y^{(0)})$ in terms of the mixing time of the Markov chain, specifically the backward transition kernel. Informally, if $v^J(y,y^\prime) \approx p(y^\prime)$, then $\int q(y)v^J(y,y^\prime) \diff y \approx p(y^\prime)$ and $\Delta^J(Y^{(0)}) \approx 1$. The departure from log-optimality of $E^J(X)$ is 
\begin{gather*}
    \EQ \bigg [\log \frac{E^J(X)}{E(X)} \bigg ] = - \E^{Q^*}[\log \Delta^J(Y^{(0)})],
\end{gather*}
or the negative e-power of $\Delta^J(Y^{(0)})$. Since $\Delta^J(Y^{(0)})$ is the likelihood ratio between $Q^*$ and $P$,
\begin{equation} \label{eq:delta-e-power}
    \E^{Q^*} [ \log {\Delta^J(Y^{(0)})} ] = \tx{KL}\left(Q^*, P\right) \leq \int q(y)\tx{KL}(V^J(y,\cdot),P) \diff y,
\end{equation}
where the inequality follows because the KL divergence is convex. For any $\epsilon>0$, define the \textit{KL-mixing time} of the backward chain to be
\begin{equation*}
    J_{\tx{KL}}(\epsilon) = \min \lb J : \max_y \KL(V^J(y,\cdot),P) < 2 \epsilon^2 \rb.
\end{equation*}
We remark $J_{\tx{KL}}(\epsilon) \geq J(\epsilon)$, where $J(\epsilon) = \min \lb J: \max_{y} d_\TV(V^J(y,\cdot),P)  < \epsilon \rb$ is the mixing time of $\Vdot$ \citep{levin2017markov}, by Pinsker's inequality. We use the inequality in equation \eqref{eq:delta-e-power} to relate the KL mixing time to the e-power of $\Delta^J(Y^{(0)})$ in the following proposition.

\begin{prop} \label{prop:mixing-time}
    Let $\epsilon>0$ and suppose that $J_{\KL}(\epsilon)<\infty$. Then, $ \E^{Q^*}[ \log {\Delta^{J_{\KL}(\epsilon)}(Y^{(0)})} ] < 2 \epsilon^2$.
\end{prop}
 
Proposition~\ref{prop:mixing-time} states that if we take $J \geq J_{\tx{KL}}(\epsilon)$ steps, then $\EQ[\log E^J(X)/E(X)] > -2 \epsilon^2.$ Combined with Theorem~\ref{thm:Besag-parallel}, using a fast-mixing MCMC algorithm and a large number of samples $M$ results in an e-variable that is approximately log-optimal. We conclude by commenting that the convergence theorem for Markov chains guarantees a finite $J(\epsilon)$ value for an irreducible and aperioidic chain \citep{levin2017markov}. If $J=J(\epsilon)$ in Algorithm~\ref{alg:besag-clifford-parallel}, Proposition~\ref{prop:delta-difference} in the supplemental material states $\int  \int q(y) |v^{J(\epsilon)}(y,y^{(0)}) - p(y^{(0)}) | \diff y \diff y^{(0)}$ is strictly less than $2\epsilon$; however, we require the stronger condition on the KL-mixing time to bound $\EQ[\log E^J(X)/E(X)]$.

\subsection{Multiple Markov chains} \label{sect:multiple-markov-chains}

A Besag-Clifford e-value $\widehat E_M(x)$ relies on running $M \geq 1$ MCMC algorithms with the same initialization point, which is $Y^{(0)}$. If the MCMC algorithm has poor mixing, $J$ is small, or $Y^{(0)}$ is in the tails of $P$, then the power of $\textbf{1}(\Ehat_M(X) \geq 1/\alpha)$ may decline since, ideally, we could compare $X$ to iid samples from $P$. In this section, we robustify against negative impacts to e-power by using parallel sweeps of Algorithm~\ref{alg:besag-clifford-parallel}, taking inspiration from MCMC practices in Bayesian inference  \citep{sountsov2024running,margossian2025nested}. Let $Y^{(0)}_1, \dots, Y^{(0)}_S$ be a list of $J$-step backward evolutions of the MCMC chain initialized at $X$. If $X \sim P$, $Y^{(0)}_s \sim P$, and if $X \sim Q$, then $Y^{(0)}_s \sim  Q^* $. Then, let $\Ehat_{M,s}(x)$ be a Besag-Clifford e-value (equation \eqref{eq:MC-test}) computed using independent $J$-step forward transitions from $Y_s^{(0)}$ and test statistic $T_s(X)$. Define
\begin{gather}
    \bar E_{M,S}(X) = \frac{1}{S} \sum_{s=1}^S \Ehat_{M,s}(X); \label{eq:indpenendent-MCMC-avg}
\end{gather}
$\bar E_{M,S}(X)$ is simply the average of the Besag-Clifford e-variables across the chains. There are two key properties of $\bar E_{M,S}(X)$: validity and increased e-power.

\begin{prop} \label{prop:e-bar}
    For any non-negative test statistics $T_s(x)$ for $s=1, \dots, S$,
    \begin{itemize}
        \item[(a)] $\bar E_{M,S}(x)$ is a valid e-value for all $M,S \in \mathbb N$;
        \item[(b)] If $T_s(x) = T(x)$ for all $s=1, \dots, S$, then $\EQ[\log \bar E_{M,S}(X)] \geq \EQ[\log \Ehat_M(X)]$ for all $M,S \in \mathbb N$.
    \end{itemize}
    \begin{proof}
        Part (a) is a fundamental property of e-values:
        \begin{gather*}
            \EP[\bar E_{M,S}(X)] = (1/S)\sum_{s=1}^S \EP[\Ehat_{M,s}(X)] \leq (1/S) \cdot S = 1.
        \end{gather*}
        Then $P(\bar E_{M,S} \geq 1/\alpha) \leq \alpha$. For (b), concavity of $\log(z)$ gives that
        \begin{equation*}
            \log \bar E_{M,S}(X) \geq \frac{1}{S} \sum_{s=1}^S \log \Ehat_{M,s}(X)
        \end{equation*}
        almost surely. The result follows because if $T_s(x) = T(x)$, and the same MCMC algorithm is used for all $s$, $\Ehat_{M,s}(X)$ are identically distributed.
    \end{proof}
\end{prop}

Proposition~\ref{prop:e-bar}(b) states that we increase the e-power of our initial e-variable $\Ehat_M(X)$ by using multiple chains, whereas part (a) confirms type-1 error protection for $\textbf{1}(\bar E_{M,S}(X) \geq 1/\alpha)$; both guarantees apply for any $M$ and $S$. In other words, averaging over multiple chains boosts the e-power of the initial Besag-Clifford e-value without sacrificing validity. In addition, one could use different MCMC algorithms $U_s(\cdot, \cdot)$, numbers of steps $J_s$, or numbers of samples $M_s$, and Proposition~\ref{prop:e-bar}(a) would still hold due to the exchangeability guarantee in \cite{besag1989generalized}. A further consequence of Proposition~\ref{prop:e-bar} is that we can define a \textit{running average LRT}, where we produce $\bar E_{M,1}, \bar E_{M,2}, \dots,$ by running sequential MCMC chains and rejecting $\mathcal P$ at the smallest $S$ so that $\bar E_{M,S} \geq 1/\alpha$ \citep{ramdas2025hypothesis}. 

In the case where we implement parallel runs of the same MCMC algorithm, we cannot achieve log-optimality exactly. Though Proposition~\ref{prop:e-bar} guarantees an increase in the e-power of $\bar E_{M,S}(X)$ over that of $\Ehat_M(X)$, autocorrelation in $\Udot$ leads to a departure from log-optimality, even when we observe infinitely many MCMC samples across infinitely many chains.

\begin{prop} \label{prop:Besag-multi-chain}
    Under the same assumptions as Theorem~\ref{thm:Besag-parallel}, as $S \to \infty$,
    \begin{equation} \label{eq:parallel-pr-limit-multi-chain}
        \frac{1}{S} \sum_{s=1}^S \frac{E(X)}{\Delta^J(Y_s^{(0)})} \overset{pr.} {\longrightarrow} E(X) \E \bigg [ \frac{1}{\Delta^J(Y^{(0)})} \bigg | X \bigg].
    \end{equation}
\end{prop}

As shown in equation \eqref{eq:delta-rewritten}, $\E[1/\Delta^J(Y^{(0)}) \mid X]$ is an exact e-variable for $Q$ and can be interpreted in terms of the mixing: if $v^J(y,\cdot) \approx p(y)$, then $\E[1/\Delta^J(Y^{(0)}) \mid X] \approx 1$. Moreover, $\E[1/\Delta^J(Y^{(0)}) \mid X]$ does not depend on $Y^{(0)}$, removing sensitivity to the initial backward steps in Algorithm~\ref{alg:besag-clifford-parallel}. By Jensen's inequality and the law of total expectation, $\EQ[\log\E[1/\Delta^J(Y^{(0)}) \mid X]] \geq - \E^{Q^*}[\log\Delta^J(Y^{(0)})]$, i.e., adding more chains pushes the probability limit closer to log-optimality from the single chain case.

In conclusion, the parallel method of \cite{besag1989generalized} with $T(x) \propto E(x)$ results in approximate log-optimality as $M \to \infty$. The mixing of the MCMC algorithm is paramount for the e-power at large values of $M$; if $J$ is set to the KL-mixing time, then in the large $M$ limit our procedure is log-optimal up to an arbitrary positive constant. We can improve the e-power of $\Ehat_M(X)$ via parallel MCMC runs, though this alone is not sufficient to ensure log-optimality, even as $S \to \infty$.

\subsection{Example: autoregressive models} \label{section:ar1}

Consider the autoregressive model of order $1$, or AR(1) model,
\begin{equation} \label{eq:AR1-model}
    Y^{(j)} = \phi Y^{(j-1)} + \gamma \epsilon^{(j)},
\end{equation}
where $\epsilon^{(j)} \sim \mathcal N(0,1)$, $\gamma>0$, $|\phi| < 1$, and $\epsilon^{(j)} \ind Y^{(k)}$ for all $j,k \in \mathbb N$. The forward transition kernel is $U(y, \cdot) = \N(\phi y, \gamma^2 )$, and the stationary distribution is $P = \N(0,\gamma^2/\{1 - \phi^2\})$. Under these conditions, the AR(1) process is irreducible, aperiodic, and reversible, i.e., $V(y,y^\prime) = U(y,y^\prime) = \N(\phi y,\gamma^2)$. Assume $\gamma^2 = 1 - \phi^2$, meaning that the stationary distribution is $P=\N(0,1)$ and the $J$-step transition kernel is $U^J(y,\cdot) = \N(\phi^J y, 1-\phi^{2J})$. Lastly, take the alternative distribution to be $Q = \N(\mu,\sigma^2)$ for some $\mu \in \mathbb R$ and $\sigma^2>0$. We re-derive the results from the previous sections under two testing scenarios, where the AR$(1)$ process is used in Algorithm~\ref{alg:besag-clifford-parallel} as an MCMC sampler.

\subsubsection{Mean-shift}
First, we fix $\mu \in \mathbb R$ and $\sigma^2 = 1$ as a test for \textit{mean-shift}. The likelihood ratio for these hypotheses is
\begin{equation} \label{eq:AR1-mean-LR}
    E(x) = \exp \lb \mu x - \frac{\mu^2}{2} \rb.
\end{equation}
The moment generating function (MGF) of a normal distribution yields the following expression for $\Delta^1(Y^{(0)})$,
\begin{equation} \label{eq:AR1-mean-Delta}
    \Delta^1(y^{(0)}) = \exp \lb  \phi  \mu y^{(0)} - \frac{\phi^2 \mu^2}{2} \rb \tx{ for } \mu \in \mathbb R.
\end{equation}
$\Delta^1(y^{(0)}) \in (0,\infty)$ for all $y^{(0)} \in \mathbb R$ and $\Delta^1(y^{(0)})$ is continuous, so Theorem~\ref{thm:Besag-parallel} is an almost-sure guarantee. As $(\gamma^2, \phi^2) \to (1,0)$, then $\Delta^1(y^{(0)}) \to 1$, and $E^1(X) \to E(X)$ almost surely. In this special limiting case, equation \eqref{eq:AR1-model} is simply independent sampling from the $\N(0,1)$ distribution, and so Theorems~\ref{thm:ergodic-log-optimal-iid} and~\ref{thm:Besag-parallel} are equivalent statements. Equation \eqref{eq:AR1-mean-LR} shows that $\Delta^1(y^{(0)})$ is the likelihood ratio for testing $\mathcal P = \{ \N(0,1) \}$ against $\mathcal Q = \{ Q^* \}$, where $Q^* = \N(\phi \mu, 1)$. We can also derive equation \eqref{eq:AR1-mean-Delta} from equation \eqref{eq:delta-rewritten} since $\int q(y) V(y,\cdot) \diff y = \N(\phi \mu, 1) = Q^*$.

Since $Y^{(0)} = \phi X + \gamma \epsilon^{(0)}$ with $\epsilon^{(0)} \sim \N(0,1)$, $\Delta^1(Y^{(0)})$ may be written explicitly in terms of $X$:
\begin{equation*}
    \Delta^1(X) := \Delta^1(\phi X + \gamma \epsilon^{(0)}) = \exp \lb \phi^2 \mu X + \phi \mu \gamma \epsilon^{(0)} - \frac{\phi^2 \mu^2}{2}   \rb.
\end{equation*}
We can verify that $1/\Delta^1(X)$ is an exact e-variable for $Q = \N(\mu,1)$ using MGFs and the facts that $\epsilon^{(0)} \ind X$ and $\gamma^2 = 1 - \phi^2$. 
\begin{gather*}
    \EQ\bigg [ \frac{1}{\Delta^1(X)} \bigg ] = \exp \lb \frac{\phi^2 \mu^2}{2} \rb \EQ[\exp \lb - \phi^2 \mu X \rb]\E[\exp \lb - \phi \mu \gamma \epsilon^{(0)} \rb] \\
    = \exp \lb \frac{\phi^2 \mu^2}{2} \rb  \exp \lb - \phi^2 \mu^2 + \frac{\phi^4 \mu^2}{2} \rb \exp \lb \frac{\phi^2 \mu^2 \gamma^2}{2} \rb \\
    = \exp \lb \frac{\phi^4 \mu^2 + \phi^2 \mu^2 (1-\phi^2) - \phi^2 \mu^2}{2} \rb = \exp \lb \frac{\phi^4 \mu^2 + \phi^2 \mu^2 - \phi^4 \mu^2 - \phi^2 \mu^2}{2} \rb = 1.
\end{gather*}
Likewise, we have that $\Delta^1(X)$ is an exact e-variable for $P=\N(0,1)$ using a similar derivation.
\begin{gather*}
    \EP[\Delta^1(X)] = \exp \lb - \frac{ \phi^2 \mu^2}{2} \rb \EP[\exp \lb \phi^2 \mu X \rb] \E[\exp \lb \phi \mu \gamma \epsilon^{(0)} \rb] \\ 
    = \exp \lb - \frac{ \phi^2 \mu^2}{2} \rb \exp \lb \frac{\phi^4 \mu^2}{2} \rb \exp \lb \frac{\phi^2 \mu^2 \gamma^2}{2} \rb \\
    = \exp \lb  \frac{\phi^4 \mu^2 + \phi^2 \mu^2 (1-\phi^2) - \phi^2 \mu^2}{2}\rb = \exp \lb \frac{\phi^4 \mu^2 + \phi^2 \mu^2 - \phi^4 \mu^2 - \phi^2 \mu^2}{2} \rb = 1.
\end{gather*}

For general $J \in \mathbb N$, $\Delta^J(y^{(0)}) = \exp \lb \phi^J \mu y^{(0)} - (\phi^{2J} \mu^2)/2 \rb$ and
\begin{equation} \label{eq:AR1-mean-Delta-j}
     \Delta^J(X) = \exp \lb \phi^{2J} \mu X + \phi^J \mu \sqrt{1-\phi^{2J}} \delta^{(0)} - \frac{\phi^{2J} \mu^2}{2} \rb,
\end{equation}
where $\delta^{(0)} \ind X$ and $\delta^{(0)} \sim \N(0,1)$. Note that this is the same formula as $\Delta^1(Y^{(0)})$ but with $\phi^J$ in place of $\phi$. Therefore, the same proofs show that $\Delta^J(X)$ is an exact e-variable for $P$, and $1/\Delta^J(X)$ is an exact e-variable for $Q$. We can rewrite $\Delta^J(X)$ in terms of the likelihood ratio, as $\Delta^J(X) = E(X)^{\phi^{2J}} e^{ \phi^J \mu \sqrt{1-\phi^{2J}} \delta^{(0)}}$,
or, alternatively,
\begin{equation} \label{eq:AR1-mean-Delta-normal}
    \log \Delta^J(X) = \phi^{2J} \log E(X) + \phi^J \mu \sqrt{1-\phi^{2J}} Z \tx{ where } Z \ind X, \: Z \sim \N(0,1).
\end{equation}
Equation \eqref{eq:AR1-model} shows that $\log \Delta^J(X) \mid X$ is a normal distribution and that $\E^Q[\log \Delta^J(X)] = \phi^{2J} \KL(Q,P) = \phi^{2J} \mu^2/2$, so $E^J(X)$ is log-optimal up to a term that is proportional to $\phi^{2J}$. We can use equation \eqref{eq:AR1-mean-Delta-j} to show that
\begin{equation*}
    \E \bigg [ \frac{1}{\Delta^J(X)} \bigg | X \bigg] = \exp \lb - \phi^{2J} \mu X + \phi^{2J} \mu^2 - \frac{\phi^{4J} \mu^2}{2} \rb,
\end{equation*}
which is the multiplicative remainder in Proposition~\ref{prop:Besag-multi-chain}. $\EQ[\log \E[1/\Delta^J(X) \mid X]] = - \phi^{2J} \mu^2 + \phi^{2J} \mu^2 - \phi^{4J} \mu^2/2 = - \phi^{4J} \mu^2/2$, and so increasing the number of chains decreases the e-power by a factor of $\phi^{2J}$, in the infinite chain and iteration setting. With regard to Proposition~\ref{prop:mixing-time}, $\int q(y) V^J(y,\cdot) \diff y = \N(\phi^J\mu, 1)$, and
\begin{equation*}
    \KL(\N(\phi^J \mu, 1), \N(0,1)) = \frac{\phi^{2J} \mu^2}{2} = O(\phi^{2J}).
\end{equation*}
Then $\EQ[\log \Delta^J(Y^{(0)})] = O(\phi^{2J})$, as we showed in equation \eqref{eq:AR1-mean-Delta-normal}. 
For fixed $\phi$, setting $J$ large enough will result in approximate log-optimality, i.e., $J = \min \{j \in \mathbb N:  \EQ[\log \Delta^J(Y^{(0)})] < 2 \epsilon^2\}$ is finite for all $\epsilon>0$, and is given by $J = \min \{j \in \mathbb N : \phi^{2J} < 4 \epsilon^2/\mu^2\}$. We demonstrate the above results through simulations in Section~\ref{sect:ar1-sims}.

\subsubsection{Rescaling}

Next, we set $\mu = 0$ and $\sigma^2 \neq 1$, effectively testing the \textit{scale} of the data generating process. The likelihood ratio for this problem is
\begin{equation} \label{eq:AR1-LR}
    E(x) = \frac{1}{\sigma} \exp \lb - \frac{x^2(1-\sigma^2)}{2 \sigma^2} \rb.
\end{equation}
Using the property that the square of a standard normal is a chi-squared distribution,
\begin{gather} \label{eq:AR1-Delta}
    \Delta^1(y^{(0)}) = \frac{1}{\sigma\sqrt{1 + \frac{\gamma^2}{\sigma^2}(1-\sigma^2)}} \exp \lb - \frac{\frac{\phi^2 y^{(0)2}}{2\sigma^2}(1-\sigma^2)}{1 + \frac{\gamma^2}{\sigma^2}(1-\sigma^2)} \rb \tx{ for } \frac{1-\sigma^2}{\sigma^2} > - \frac{1}{\gamma^2}.
\end{gather}
By definition, $\Delta^1(y^{(0)}) \geq 0$ for all $y^{(0)} \in \mathbb R$. If $\sigma^2 \leq 1$, then $\Delta^1(y^{(0)})<\infty$, otherwise, we must choose $\gamma^2$ to agree with the condition in equation \eqref{eq:AR1-Delta}. Also, note that $\Delta^1(y^{(0)})$ is a continuous function of $y^{(0)}$, provided that the condition on $\gamma^2$ holds, so the convergence statement in Theorem~\ref{thm:Besag-parallel} holds as an almost-sure statement.

The formula in equation \eqref{eq:AR1-LR} implies that $\Delta^{1}(y^{(0)})$ is the likelihood ratio between $P = \N(0,1)$ and $Q^* = \N(0,\nu^2)$, where $\nu^2 = \sigma^2\{1+(\gamma^2/\sigma^2)(1-\sigma^2)\}$, and recall that $\phi^2 = 1 - \gamma^2$. By equation \eqref{eq:AR1-model} and reversibility, $\int q(y)V(y,\cdot) \diff y = \N(0, \phi^2 \sigma^2 + \gamma^2) = \N(0,\nu^2) = Q^*$, since $\phi^2 \sigma^2 + \gamma^2 = \nu^2$. 
Finally, $\int q(y)V^J(y,\cdot) \diff y = \N(0, \phi^{2J} \sigma^2 + 1 - \phi^{2J})$. Then,
\begin{equation*}
    \KL(\N(0, \phi^{2J} \sigma^2 + 1 - \phi^{2J}), \N(0,1)) = \frac{1}{2} \lb \phi^{2J}( \sigma^2 - 1) - \log(\phi^{2J} \sigma^2 + 1 - \phi^{2J})\rb = O(\phi^{2J}),
\end{equation*}
and so there exists a finite $J$ satisfying approximate log-optimality for all $\epsilon>0$, similar to the mean-shift case. 

\section{Beyond simple hypotheses} \label{section:composite}

In practice, the hypotheses of interest are often not simple. Take, for instance, a parametric model defined by a parameter $\theta \in \mathbb R^n$. For testing $\mathcal P = \{ \theta = 0\}$ against $\mathcal Q = \{ \theta \neq 0\}$, the alternative hypothesis encompasses possibly infinitely many models. In this section, we establish a method for constructing Besag-Clifford e-values that are valid under composite hypotheses. As before, we do not rely on asymptotic guarantees for type-1 error control.

\subsection{Composite nulls and confidence regions}

If $\mathcal P$ is composite but finite, i.e., $\mathcal P = \{P_1, \dots, P_R \}$, but we only have access to $p_r(x) \propto g_r(x)$, we can design an e-variable using minimization. For each null, we set $T_r(X) = q(X)/g_r(X)$ and  simulate $Y^{(1)}_r, \dots, Y_r^{(M)}$ via Algorithm~\ref{alg:besag-clifford-parallel} using an MCMC algorithm that is stationary for $P_r$. Then, let
\begin{equation} \label{eq:finite-null}
    \Ehat_M(X) = \min_{r=1, \dots, R} \widehat E_{M,r}(X), \tx{ where } \widehat E_{M,r}(X) = \frac{(M+1)T_r(X)}{T_r(X) + \sum_{m=1}^M T_r( Y^{(m)}_r)};
\end{equation}
or the minimum value attained over the Besag-Clifford e-variables applied to each null.
\begin{prop} \label{prop:composite-snlr}
    The composite Besag-Clifford e-value in equation \eqref{eq:finite-null} is valid.
\end{prop}
\begin{proof}
    By definition,
    \begin{gather*}
        \E^{P_r}[\Ehat_M(X)] \leq \E^{P_r} \bigg [ \frac{(M+1)T_r(X)}{T_r(X) + \sum_{m=1}^M T_r(Y^{(m)}_r)} \bigg] \leq 1, 
    \end{gather*}
    since, if $P_r$ is true, then $X, Y_r^{(1)}, \dots, Y_r^{(M)}$ are exchangeable.  
\end{proof}

As in the simple null case, we can construct an analogous e-variable using equation \eqref{eq:finite-null} if $q(x) \propto h(x)$. Recall that we generally cannot maximize the likelihood exactly without knowing the normalizing constant for each $P_r$. The minimization serves as an approximate maximization of the likelihood, as the denominator converges to $K_r \Delta^{J_r}(Y_r^{(0)})$ with increasing $M$ for each $r \in \{1, \dots, R \}$.

Similarly, we can define a confidence region for parametric models. In this context, $\mathcal P = \{ P_\theta: \theta \in \Theta \}$ is a distribution family, $\Theta= \{ \theta_1, \dots, \theta_R\}$ is a parameter space, and the e-value and MCMC samples now depend on $\theta$. Define
\begin{equation} \label{eq:confidence-region}
    \mathcal B_M(X) = \lb \theta \in \Theta: \widehat E_{M,\theta}(X) < \frac{1}{\alpha} \rb.
\end{equation}
Then $\mathcal B_M(X)$ is a valid confidence region for $\theta$, since
\begin{equation} \label{eq:valid-confidence-interval}
    P_\theta(\theta \not \in \mathcal B_M(X)) = P_\theta\{\widehat E_{M,\theta}(X) \geq 1/\alpha\} \leq \E^{P_\theta}[\widehat E_{M,\theta}(X)] \leq \alpha.
\end{equation}
Algorithm~\ref{alg:besag-clifford-cr} is a simple method for computing a \textit{Besag-Clifford confidence region}. Note that the region $\mathcal B_M(X)$ is randomized (i.e., it depends on the MCMC samples), though we can stabilize our inferences and decrease the expected interval width by setting $M$ large and averaging over multiple MCMC chains, as introduced in equation \eqref{eq:indpenendent-MCMC-avg}. The coverage guarantee does not depend on $\Theta$ being finite, but the feasibility of Algorithm~\ref{alg:besag-clifford-cr} does, since we cannot produce MCMC samples for infinitely many stationary distributions. 

We emphasize that a Besag-Clifford confidence region differs from a Bayesian credible interval, though both use MCMC samples for uncertainty quantification. In the latter case, we run a single MCMC algorithm to generate samples from some unnormalized probability distribution and express uncertainty using the empirical distribution of these samples, e.g., the quantiles \citep{hoff2009first}. Our approach runs a separate MCMC algorithm for each $\theta \in \Theta$, sacrificing computational efficiency in favor of validity (equation \eqref{eq:valid-confidence-interval}). However, Algorithm~\ref{alg:besag-clifford-cr} is applicable for parallelization, where chains that are stationary for each member of $\mathcal P$ are run simultaneously \citep{sountsov2024running,margossian2025nested}.

\begin{algorithm}[t]
\caption{Besag-Clifford Confidence Region}
\label{alg:besag-clifford-cr}
\begin{algorithmic}[1]
\Require{Stationary MCMC algorithms for $P_\theta$, empty confidence set $\mathcal B_M(X) = \emptyset$.}
    \For{$\theta \in \Theta$} 
        \State Calculate $\widehat E_{M,\theta}(X)$ using Algorithm~\ref{alg:besag-clifford-parallel};
        \If {$\widehat E_{M,\theta}(X) < 1/\alpha$} $\theta \in \mathcal B_M(X)$
        \EndIf
    \EndFor
\State \Return $\mathcal B_M(X)$.
\end{algorithmic}
\end{algorithm}

\subsection{Composite alternatives}

If $\mathcal Q$ is composite, we can use our methodology to create a level-$\alpha$ test by setting $T(x) \propto \hat q(x)/p(x)$, where $\hat q$ is an estimator of the distribution of $X$ under the alternative hypothesis, e.g., setting $\hat q(x) = \prod_{i=1}^n \hat q_{i-1}(x_i)$, where $\hat q_{i-1}$ is an estimate of $q$ using $x_1, \dots, x_{i-1}$ \citep{wasserman2020universal}. Our methodology is also applicable to $\hat q(x) = \prod_{i=1}^n \hat q_{n}(x_i;x)$, where $\hat q_n(x_i;x)$ is an estimate of $q$ using the full dataset. Even if the normalizing constant of $p(x)$ is known, $T(x) = \{ \prod_{i=1}^n \hat q_n(x_i;x) \}/p(x)$ is not an e-value, but has null expectation $\EP[T(X)] = \int \prod_{i=1}^n \hat q_n(x_i;x) \diff x \neq 1$. Proposition~\ref{prop:MC-test} applies here, meaning that we can always design a level-$\alpha$ test with this choice of $T(x)$. If this integral is finite, $\hat E_M(X)$ converges in probability to the appropriately normalized e-variable under iid sampling, or a scaled version of that e-variable if the parallel method is used to sample $P$, by Theorem~\ref{thm:ergodic-log-optimal-iid} and~\ref{thm:Besag-parallel}, respectively. A special case is when $T(x)$ is already an e-value, but $\EP[T(X)] < 1$. If $Y^{(m)} \iid P$, then $\widehat E_M(X) \overset{a.s.}{\to} T(X)/\E^P[T(X)] > T(X)$, so the Besag-Clifford e-value has higher e-power than the original e-value for arbitrarily large $M$.

\section{Besag-Clifford e-processes} \label{section:sequential}

Say that we observe data $X_1, X_2, \dots $ sequentially and would like to construct a likelihood ratio testing procedure that is valid at any (possibly data-dependent) stopping time. \textit{E-processes} are stochastic processes of e-variables that result in anytime-valid tests \citep{wasserman2020universal,ramdas2025hypothesis}. Initially, we assume $\mathcal P = \{P\}$ is a simple null, which, in the sequential context, refers to observing an iid sequence $X_1, X_2, \dots, $ from $P$. For $t>0$, let $Y_t^{(1)}, \dots, Y_t^{(M_t)}$ be simulated from $X_t$ using Algorithm~\ref{alg:besag-clifford-parallel}, \textit{independently} across time. Then, define
\begin{align} \label{eq:sequential-snlr}
    \widehat U_t(X_t) & = \frac{(M_t+1) T(X_t)}{T(X_t) + \sum_{m=1}^{M_t} T( Y_t^{(m)})}; & \Ehat_t & = \prod_{i=1}^t \{ 1 - \lambda_{i-1} + \lambda_{i-1} \widehat U_i(X_i) \}),
\end{align}
where $\lambda_{i-1}$ is measurable with respect to $\sigma(\widehat U_{i}(X_i): i<t)$, and $\widehat U_t(x_t)$ is a Besag-Clifford e-value. We refer to $\widehat E_t$ as a \textit{Besag-Clifford e-process}, which we justify below.

\begin{prop} \label{prop:sequential-snlr}
    Let $\Ehat_t$ be defined in equation \eqref{eq:sequential-snlr} for all $t>1$, where $\Ehat_0=1$ and $\lambda_0 \in [0,1]$ is fixed. Then $\Ehat_t$ is an e-process for all $M_1, \dots, M_t \in \mathbb N$ with respect to the filtration $\mathcal F_{t} = \sigma(\widehat U_i(X_i):i\leq t)$.
\end{prop}
\begin{proof}
    Since $\lambda_{i-1}$ is measurable with respect to $\mathcal F_{t-1}$, 
    \begin{equation*}
        \EP[1-\lambda_{t-1} + \lambda_{t-1}\widehat U_t(X_t) \mid \mathcal F_{t-1}] = 1 - \lambda_{t-1} + \lambda_{t-1}\E[\widehat U_t(X_t) \mid \mathcal F_{t-1}] \leq 1 - \lambda_{t-1} + \lambda_{t-1} = 1.
    \end{equation*}
    It follows that $\EP[\Ehat_t \mid \mathcal F_{t-1}] = \Ehat_{t-1} \EP[1-\lambda_{t-1} + \lambda_{t-1}\widehat U_t(X_t) \mid \mathcal F_{t-1}] \leq \Ehat_{t-1}$. Therefore, $E_t$ is a test-supermartingale \citep{grunwald2024safe} and, consequently, an e-process.
\end{proof}

The formulation of $\widehat E_t$ is general, and we need not have equal sample sizes $M_t$ or equivalent MCMC algorithms across time in order to ensure validity. For a simple alternative, a special case of a Besag-Clifford e-process is $\lambda_i = 1$ for all $i>0$, resulting in $\Ehat_t = \prod_{i=1}^t \widehat U_t(X_t)$. This e-process is comparable to the likelihood ratio (LR) process $E_t = \prod_{i=1}^t E(X_t)$, and $E_t$ attains an optimal growth rate for sequential testing of simple hypotheses, equal to $\KL(Q,P)$ \citep{ramdas2025hypothesis}. We refer to $\Ehat_t = \prod_{i=1}^t \widehat U_t(X_t)$ as a \textit{ULR process} since it consists of ULRs that have been normalized via the Besag-Clifford samples. Below, we formalize the key property of a ULR process: convergence to the LR process. 

\begin{prop} \label{thm:sequential-besag-clifford}
    Assume there exists a $J \in \mathbb N$ so that $\Delta^J(y^\prime) \in (0,\infty)$ and is continuous for all $y^\prime \in \mathbb R^n$. For all $t \in \mathbb N$, as $M \to \infty$,
    \begin{equation} \label{eq:sequential-besag-clifford}
        \widehat E_t \overset{a.s.}{\longrightarrow} \frac{E_t}{\prod_{i \leq t }\Delta^J(Y_i^{(0)})}.
    \end{equation}
\end{prop}

The proof is a trivial application of Theorem~\ref{thm:Besag-parallel}. As in the non-sequential case, there is a multiplicative error term in the limit of $\widehat E_t$, denoted $\Delta^J_t: = \prod_{i=1}^t \Delta^J(Y_i^{(0)})$. $\Delta_t^J $ and $1/\Delta_t^J$ are e-processes for $P$ and $Q$ respectively with respect to the filtration $\sigma(Y_i^{(0)}:i \leq t)$. If the alternative is true, then $Y_i^{(0)} \iid Q^*$ for all $i$, and so by the strong law of large numbers, $(1/t)\sum_{i=1}^t\log \Delta^J(Y_i^{(0)}) \overset{a.s.}{\to}\KL(Q^*,P)$, which is negligible for well-mixing chains by equation \eqref{eq:delta-e-power}. Thus, $E_t/\Delta_t^J$ achieves a near-optimal asymptotic growth rate; note that $E_t/\Delta_t^J$ is an e-process by Fatou's lemma applied to Theorem~\ref{thm:Besag-parallel}; we formalize this result below.

\begin{prop}
    Assume that $X_i \sim Q$ for $i=1, \dots, t$. Then, as $t \to \infty$,
    \begin{equation*}
       \frac{\log E_t - \log \Delta_t^J}{t} \overset{a.s.}{\longrightarrow} \KL(Q,P) - \KL(Q^*,P).
    \end{equation*}
\end{prop}

Our methodology can be extended to the composite null setting by applying the product formula in equation \eqref{eq:sequential-snlr} with $\widehat U_i(X_i)$ being the composite Besag-Clifford e-values introduced in Section~\ref{section:composite}. If all MCMC sampling is carried out independently across time points, Proposition~\ref{prop:sequential-snlr} holds, as $\widehat U_i(X_i)$ is an e-variable conditional on $\widehat U_0(X_0), \dots, \widehat U_{i-1}(X_{i-1})$. We can also integrate multiple chains using the same strategy as equation \eqref{eq:indpenendent-MCMC-avg} and preserve the e-process property. Lastly, Besag-Clifford e-processes are compatible with other standard betting strategies, i.e., methods to select $\lambda_t$. Examples include growth rate adaptive to the particular alternative (GRAPA), approximate GRAPA (aGRAPA), and lower-bound on the wealth (LBOW); see \cite{waudbysmith2023estimating} for technical details on these methods. Sections~\ref{sect:sequential-product-of-experts} and~\ref{sect:sequential-composite-alt} provide simulated examples of the above Besag-Clifford e-processes.

\section{Additional examples} \label{section:examples}

\subsection{Restricted Boltzmann machines}
A key application of our work is testing in restricted Boltzmann machines (RBMs) \citep{fischer2012introduction}, including product of experts (PoE) \citep{hinton1999products} and Markov random fields (MRFs) \citep{kindermann1980markov}. RBMs are expressed in terms of an unnormalized energy function, such as $p(x;\theta) \propto \exp \{ - \mathcal E(x;\theta)\}$, thus a ULR is given by $T(X) = \hat q(X)\exp \{ \mathcal E(X;\theta)\}$. 

The PoE model assumes that
\begin{equation} \label{eq:product-of-experts}
    p(x_i;\theta) \propto \prod_{w=1}^W g_w(x_i;\theta_w),
\end{equation}
where $g_w(x_i;\theta_w)$ are the unnormalized densities of a parametric family. If $g_w$ is a Gaussian distribution, then $P$ is a multivariate Gaussian. However, in practice, $g_w$ is often set to be a t-distribution \citep{welling2002learning}, in which case, $g_w(x_i;\theta_w) \propto \left( 1 + x_i^2/\theta_w \right)^{-(\theta_w+1)/2}$ and $\theta_w > 0$. Noncentrality parameters $\phi_w$ and scale parameters $\sigma_w$ can be incorporated into the PoE model as well. The normalization constant of $p(x_i;\theta)$ is not known, but the PoE model can be trained with contrastive divergence \citep{hinton2002training} or score matching \citep{hyvarinen2005estimation,cai2025fisher}. For instance, suppose we would like to test which number of experts fits the data well, say, $W$ or $W+1$ for some $W>0$. We train both models, resulting in parameters $\theta^W = (\theta^W_1, \dots, \theta^W_W)$ and $\theta^{W+1} = (\theta^{W+1}_1, \dots, \theta_{W+1}^{W+1})$. Then, using held out data $X$, we can test $H_0:W$ against $H_1:W+1$ with $T(X) \propto \prod_{i=1}^n \{ p(X_i;\theta^{W+1})/p(X_i;\theta^W) \}$. The Besag-Clifford e-variable is computed via Algorithm~\ref{alg:besag-clifford-parallel} with any MCMC algorithm that has stationary distribution $\prod_{i=1}^n P(\cdot; \theta^W)$, e.g., Gibbs sampling \citep{hinton2002training} or derivative-based methods such as MALA or HMC. 

Given an undirected graph $G$ with a set of cliques $\mathcal C$, an MRF is defined as
\begin{equation} \label{eq:MRF}
    p(x_i;G) \propto \prod_{c \in \mathcal C} g_c(x_{ic}) = \exp \lb - \sum_{c \in \mathcal C} \phi_c(x_c) \rb,
\end{equation}
where $x_{ic} = (x_{ij})_{j \in c}$ and $\phi_{c}(x_c)$ is a clique-specific energy function. If the energy functions are non-Gaussian, then the normalizing constant of $p(x_i;G)$ is unknown. However, block Gibbs sampling of an MRF is always possible due to the factorization in equation \eqref{eq:MRF}. The Besag-Clifford likelihood ratio test can be applied in this case, e.g., to test between two possible graphs or as a general goodness-of-fit test.

\subsection{Power likelihoods}

Inferences can be made robust to model misspecification via the power likelihood \citep{royall2003interpreting, holmes2017assigning}. The power likelihood ratio is $T(x) = q(x)^\eta/p(x)^\eta$, where $0 < \eta < 1$. For testing a simple null $\mathcal P = \{P\}$, $T(X)$ is an e-variable, i.e., $\EP[T(X)] \leq 1$ \citep{wasserman2020universal}.
We can create an asymptotically more powerful test using our approach. First, suppose that $Y^{(1)}, \dots, Y^{(M)} \iid P$. Theorem~\ref{thm:ergodic-log-optimal-iid} shows that as $M \to \infty$, $\Ehat_M(X) \overset{a.s.}{\to} T(X)/\EP[T(X)] \geq T(X)$, and so the soft-rank e-variable has higher e-power than $T(X)$ as $M \to \infty$. If we cannot generate iid samples from $P$, we use the parallel method, but with a slight adjustment. For this specific problem, we rephrase the hypotheses as $\mathcal P = \{P\}$ and $\mathcal Q = \{ Q^* \}$, where the density of the alternative distribution is $q^*(x) \propto q(x)^\eta p(x)^{1-\eta}$; note that this is a valid probability distribution since $T(X)$ is an e-variable. For a single chain, Theorem~\ref{thm:Besag-parallel} implies $\Ehat_M(X) \overset{pr.}{\to} T(X)/\{\EP[T(X)] \Delta^J(Y^{(0)})\},$
where the denominator is
\begin{equation*}
   \Delta^J(Y^{(0)}) = \int \frac{q(y)^\eta}{\EP[T(X)] p(y)^\eta}u^J(Y^{(0)},y) \diff y = \frac{\int q(y)^\eta p(y)^{1-\eta} v^J(y,Y^{(0)}) \diff y  } {\EP[T(X)]p(Y^{(0)})}.
\end{equation*}
If $J=J_{\KL}(\epsilon)<\infty$, Proposition~\ref{prop:mixing-time} yields 
\begin{gather*}
    \EQ[\log \Delta^J(Y^{(0)})] = \int \int \log \left( \frac{\int q^*(y)v^J(y,y^{(0)}) \diff y}{p(y^{(0)})} \right) q(y)v^J(y,y^{(0)}) \diff y \diff y^{(0)} \\
    = \tx{KL}\left( \int q(y) V^J(y,\cdot) \diff y, P \right)  - \tx{KL}\left( \int q(y) V^J(y,\cdot) \diff y, \int q^*   (y) V^J(y,\cdot) \diff y \right) \\
    %\leq \int q(y)\tx{KL}(V^J(y,\cdot),P) \diff y  %\tx{KL}\left( \int q(y) V^J(y,\cdot) \diff y, \int q^*   (y) V^J(y,\cdot) \diff y \right) \\
    %\leq \int q(y)\tx{KL}(V^J(y,\cdot),P) \diff y 
    \leq \int q(y)\tx{KL}(V^J(y,\cdot),P) \diff y < 2 \epsilon^2,
\end{gather*}
and so we asymptotically improve on the power likelihood ratio up to a positive constant. 

\subsection{E-values}

We now compare Besag-Clifford e-values to standard e-values from a nonasymptotic perspective. Initially, we assume that we can simulate $iid$ from $P$.
We first present a more general result, where we assume that $\tilde T^{(1)}, \dots, \tilde T^{(M)}$ are simulations that are \textit{independent of the data}, so their expectation when $X \sim P$ or $X \sim Q$ is static, and denoted $\E[\tilde T^{(m)}]$. In this more general notation, we define $\Ehat_M(X) = T(X)/\{T(X)/(M+1) + \sum_{m=1}^M \tilde T^{(m)}/(M+1)\}$. 

\begin{theorem} \label{thm:null-exchanged-e-power}
    Let $\tilde T^{(1)}, \dots, \tilde T^{(M)}$ be nonnegative simulations that are independent of $X$ and have common mean $ \E[\tilde T]:= \E[\tilde T^{(m)}]  \in (0,\infty)$ for all $m=1, \dots, M$. Then for any $Q \in \mathcal Q$ and $M \in \mathbb N$,
    \begin{equation} \label{eq:null-exchanged-e-power}
    \EQ[\log \Ehat_M(X)] > \EQ[\log T(X)] - \log \left( \frac{M \E[\tilde T]}{M+1} + \frac{\EQ[T(X)]}{M+1} \right).
\end{equation}
\end{theorem}

Theorem~\ref{thm:null-exchanged-e-power} provides a relationship between the behavior of $\Ehat_M(X)$ and $T(X)$ under the alternative. At this point, we have not assumed that $T(x)$ is an e-value, meaning that this relationship holds for any unnormalized nonnegative test statistic. Below we enumerate consequences of equation \eqref{eq:null-exchanged-e-power} when $T(x)$ is an e-value.

\begin{corr} \label{corr:finite-M}
Assume $T(x)$ is an e-value for $P$, $\EQ[T(X)] < \infty$ for all $Q \in \mathcal Q$, and $\tilde T^{(m)} = T( Y^{(m)})$, where $Y^{(m)} \iid P$. 
    \begin{enumerate}
        \item If $\EP[T(X)] < 1$, let $M_*^Q = \min_{M \in \mathbb N} \{ \EQ[T(X)] -1 < M(1 - \E[\tilde T]) \}$. Then,  $\EQ[\log \Ehat_M] > \EQ[\log T(X)]$ for all $M \geq M^Q_*$. 
        \item If $\EP[T(X)] < 1$ and $\sup_{Q \in \mathcal Q} \EQ[T(X)] < \infty$, then there exists an $M_* \in \mathbb N$  so that for all $Q \in \mathcal Q$, $ \EQ[\log \Ehat_M] > \EQ[\log T(X)]$ for any $M \geq M_*$.
        \item Suppose $Q \in \mathcal Q$ satisfies $\EQ[T(X)] \leq 1$. Then $\EQ[\log \Ehat_M(X)] > \EQ[\log T(X)]$ for all $M \in \mathbb N$.  
    \end{enumerate} 
\end{corr}

That is, if $\EP[T(X)] < 1$, we can improve on the e-power for large enough $M$. For simple hypotheses, we can usually specify exact e-values, e.g. from likelihood ratios. Item 3 of Corollary~\ref{corr:finite-M} is relevant in this case; we can improve on the e-power for alternatives that lead to low power in the original e-value. Note that this guarantee holds for any $M$. The interpretation of this result is that Besag-Clifford e-values boost power for weak signals, such as parameter values close to the null-alternative boundary. Domination in e-power may not hold when $\EQ[T(X)] > 1$ and $\EP[T(X)] = 1$, therefore we recommend transforming an exact e-value into a Besag-Clifford e-value if there is prior knowledge that the signal is weak.

If we use Algorithm~\ref{alg:besag-clifford-parallel} to simulate $Y^{(1)}, \dots, Y^{(M)}$, then we replace $\E[\tilde T]$ in equation \eqref{eq:null-exchanged-e-power} with $\E^{P^*}[T( Y^{(m)})]$, where $P^*$ is the distribution of $Y^{(m)}$ when $X \sim Q$, with density $p^*(y) = \int \int u^J(y,y^\prime)v^J(y^\prime,x)q(x) \diff y^\prime \diff x$. However, items $1$ and $2$ require that there exists an $M$ so that $\EQ[T(X)] - 1 < M(1 - \E^{P^*}[T(Y^{(m)})])$. This condition is more complicated than what we have stated in Corollary~\ref{corr:finite-M} for exact sampling, as $P^*$ depends on $Q$. 

Furthermore, say $T(X)$ is a test-statistic and $E^\prime(X)$ is an e-variable (possibly exact) with $\E^Q [\log T(X)] > \E^Q [\log E(X)]$. From Theorem~\ref{thm:null-exchanged-e-power}, we can see that
\begin{equation*}
    \E^Q[\log \Ehat_M(X)] - \E^Q[\log E^\prime(X)] \genq - \log \E[\tilde T] = - \log \E^P[T(X)]
\end{equation*}
for large $M$. Then, if $T(X)$ is a test statistic that we expect to have a higher log-wealth than $E^\prime(X)$ when $\mathcal Q$ is true but is not an e-variable itself, normalization through exact sampling guarantees type-1 error control and approximate improvement over $E^\prime(X)$ in e-power. Put differently, Theorem~\ref{thm:null-exchanged-e-power} verifies approximate domination in e-power for $\Ehat_M(X)$ over $E^\prime(X)$ when $\EQ[\log T(X)] > \E^Q [\log E^\prime(X)]$ and $E^P[T(X)] = 1 + \epsilon$ for some unknown and small $\epsilon>0$. Section~\ref{sect:sequential-composite-alt} provides a simulated example of the relationship between the e-powers of an exact e-value and a Besag-Clifford e-value constructed from an unnormalized $T(x)$.

\section{Synthetic data} \label{section:simulations}

\subsection{AR$(1)$ process} \label{sect:ar1-sims}

\begin{figure}[t]
    \centering
    \includegraphics[scale=0.45]{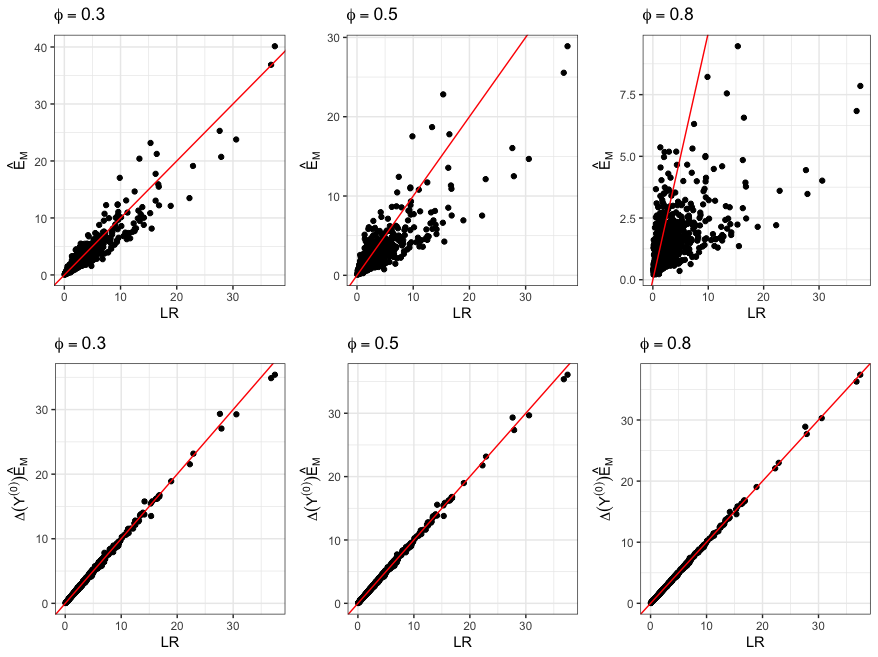}
    \caption{A demonstration of Theorem~\ref{thm:Besag-parallel} for the AR$(1)$ process. The top row compares $\Ehat_M(X)$ to $E(X)$ and the bottom row compares $\Delta^1(Y^{(0)}) \Ehat_M(X)$ to $E(X)$, where $M=1000$. The correlation is varied in $\phi \in \{  0.3, 0.5, 0.8\}$, and we simulate $1,000$ independent replications from the $\N(1,1)$ distribution.}
    \label{fig:AR1-theory}
\end{figure}

Recall the AR$(1)$ process (equation \eqref{eq:AR1-model}), which is stationary for the $\N(0,1)$ distribution when $\gamma^2 = 1 - \phi^2$. We focus on the mean-shift problem, and set $\mathcal Q = \{ \N(1,1) \}$. The one-step autocorrelation in the AR$(1)$ process is varied in $\phi \in \{ 0.3,0.5,0.8 \}$; larger values of $\phi$ increase the KL divergence between the transition kernel and the $\N(0,1)$ distribution and induce higher autocorrelation. For each choice of $\phi$, we generate $1,000$ independent replicates of $X$ under the alternative hypothesis, and set $T(x) = \exp(x)$ and $M=1,000$. 

First, we set the alternative to be the $\N(1,1)$ distribution. When the number of steps is $J=1$, $X$ and $Y^{(m)}$ are highly correlated, as the latter is generated via two forward steps from $X$ (due to reversibility), and $\Delta^1(Y^{(0)})$ is given by equation \eqref{eq:AR1-mean-Delta}.
Through repeated simulation from the alternative, $\Ehat_M(X)$ tends to underestimate $E(X)$ as the correlation increases (Figure~\ref{fig:AR1-theory}, top row). However, multiplying by $\Delta^1(Y^{(0)})$ results in near-equality between the Besag-Clifford e-value and the likelihood ratio, as suggested by Theorem~\ref{thm:Besag-parallel} (Figure~\ref{fig:AR1-theory}, bottom row).

We next fix $\phi=0.5$ and $\mu=2$ and focus on the power of $\textbf{1}(\Ehat_M(X) \geq 1/0.05)$, with $T(x) = \exp(\mu x) = \exp(2x)$. For each number of steps $J \in \{ 1,3,5,10,20 \}$, we simulate $2,500$ independent replications of $X \sim \N(2,1)$ and use these to estimate $Q(\Ehat_M(X) \geq 1/0.05)$ when $M \in \{10,50,100,500,1000,2500,5000\}$.
There is a substantial jump in power from $J=1$ to $J=3$ (Figure~\ref{fig:log-pow-ar1}), though the degree of improvement of the power diminishes as $J$ increases, which is expected due to the exponential decay of $\log \Delta^J(X)$ in equation \eqref{eq:AR1-mean-Delta-normal}. $\textbf{1}(\Ehat_M(X) \geq 1/\alpha)$ becomes more powerful for increasing $M$ (stabilizing for $M>100$) and resembles the power of $\textbf{1}(E(X)\geq 1/\alpha)$.

\begin{figure}[t]
    \centering
    \includegraphics[width=0.7\linewidth]{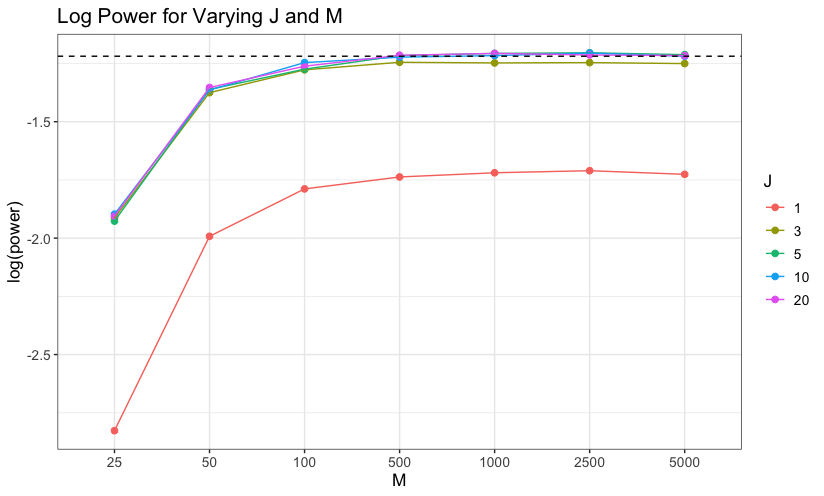}
    \caption{Estimated power (with logarithms taken for visualization) of $\textbf{1}(\Ehat_M(X) \geq 1/0.05)$ for different choices of the number of steps ($J$) and samples ($M$) in Algorithm~\ref{alg:besag-clifford-parallel} and the AR$(1)$ process. The dashed horizontal line is the estimate of $Q(E(X)\geq 1/0.05)$.}
    \label{fig:log-pow-ar1}
\end{figure}

\subsection{Sequential product of experts} \label{sect:sequential-product-of-experts}
Consider the simple hypotheses $Q = \N(\mu, \sigma^2)$ and $P$ being a product of $W=2$ noncentral t-distributions (equation \eqref{eq:product-of-experts}). At time $i$, we calculate a ULR
\begin{equation*}
    T(x_i) = \frac{e^{-\frac{(x_i-\mu)^2}{2}}}{\sqrt{2\pi \sigma^2}} \prod_{w=1}^2 \left( 1 + \frac{(x_i-\psi_w)^2}{\theta_w} \right)^{\frac{\theta_w+1}{2}}.
\end{equation*}
Over $500$ independent replications, we simulate $X_i \sim \N(\mu, \sigma^2)$ for $i=1, \dots, n=50$ time points. We compute the ULR e-process for $T(x)$ using a random-walk Metropolis algorithm, available in the \texttt{MCMCpack} R package \citep{martin2011mcmcpack}. Random-walk Metropolis is reversible, thus a single Besag-Clifford sample is obtained via applying the sampler twice (first initialized at $X$, then at $Y^{(0)}$). Motivated Proposition~\ref{prop:e-bar}, we set $J=4$, $M=25$, and $S \in \{ 1,4,10 \}$, the latter variable being the number of chains in equation \eqref{eq:indpenendent-MCMC-avg}. For the null hypothesis, we set $\theta_1 = 1$, $\theta_2 = 10$, $\psi_1 = -3$, and $\psi_2 = 0$, resulting in an asymmetric density centered at $0$, whereas the alternative is $\mu=0$, $\sigma^2=1$. 
Under the alternative, the ULR e-process grows exponentially (Figure~\ref{fig:ulr-process-t-product}). Increasing the number of chains always increases the rate of $\log \widehat E_t$ compared to the single chain e-process. 
However, the rate of the ULR e-process for $S=10$ is slower than the rate for $S=4$, suggesting that the growth rate is not always monotonic in $S$. This behavior may indicate that the running average LRT (Section~\ref{sect:multiple-markov-chains}) could oscillate as the number of chains is increased.

\begin{figure}[t]
    \centering
    \includegraphics[width=0.9\linewidth]{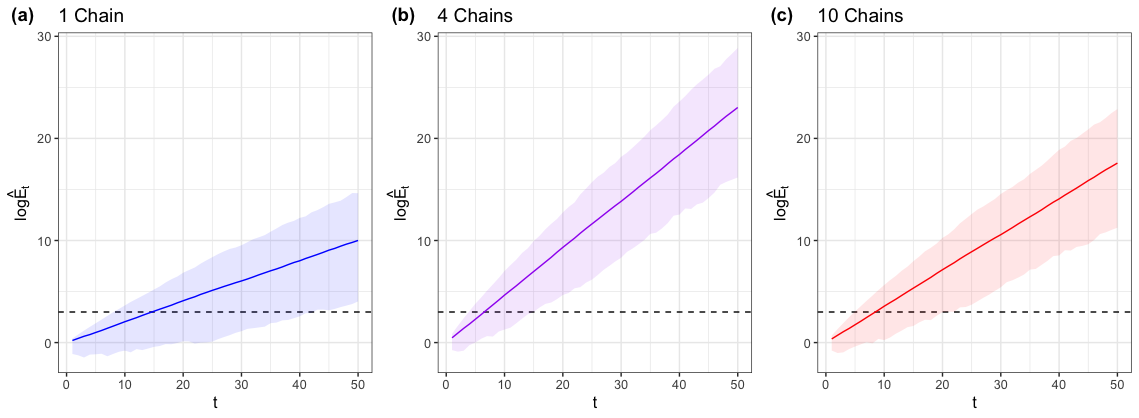}
    \caption{Averages and 95\% confidence bands of $500$ replications of the Besag-Clifford ULR e-process under the alternative, $\N(0,1)$, for $n=50$ time points. We increase the number of chains $S \in \{1,4,10\}$. The dotted line shows $\log(1/0.05)$.}
    \label{fig:ulr-process-t-product}
\end{figure}

\subsection{Sequentially testing a composite alternative hypothesis} \label{sect:sequential-composite-alt}

We return to the problem of inferring Gaussian parameters but with a composite alternative: set $\mathcal P = \{ \N(0,1) \}$ and $\mathcal Q = \{ \N(\mu, \sigma^2): \mu \in \mathbb R, \sigma^2>0\}$. At the $t$th time point, let
\begin{equation} \label{eq:unnormalized-sequential-normal}
    T(x_t \mid x_{1:(t-1)}) = \frac{\N(x_t; \bar x_t, \hat \sigma_t^2)}{\mathcal N(x_t; 0, 1)},
\end{equation}
where $\bar x_t = (1/t)\sum_{i=1}^t x_i$, $\hat \sigma_t^2 = (1/t)\sum_{i=1}^t(x_i-\bar x_t)^2$, and $\N(z; \mu, \sigma^2)$ is the pdf of the $\N(\mu, \sigma^2)$ distribution evaluated at $z$. $T(x_t \mid x_{1:(t-1)})$ is not an e-value, and hence its sequential product is not an e-process. Instead, we use $T(x_t \mid x_{1:(t-1)})$ as the test statistic to define a Besag-Clifford e-process, $\widehat E_t = \prod_{i=1}^t \{ 1 - \lambda_{i-1} + \lambda_{i-1} \widehat U_i(X_i) \}$. At each time point, we update the Besag-Clifford e-process by simulating from the $\N(0,1)$ distribution for $M=1000$ iterations via exact Monte Carlo simulation (independently of $X_t$). For comparison, let
\begin{equation*}
    D(x_t \mid x_{1:(t-1)}) = \frac{\N(x_t; \bar x_{t-1}, \hat \sigma_{t-1}^2)}{\mathcal N(x_t; 0, 1)},
\end{equation*}
meaning that $D_t = \prod_{i=1}^t \{ 1-\lambda_{i-1} + \lambda_{i-1} D(X_i \mid X_{1:(i-1)}) \}$ is an e-process \citep{wasserman2020universal}; in this example, we refer to $D_t$ as the LR e-process. For both methods, we use the GRAPA betting strategy \citep{waudbysmith2023estimating,ramdas2025hypothesis} as a method for selecting $\lambda_{i-1}$.

\begin{figure}[t]
    \centering
    \includegraphics[width=1.0\linewidth]{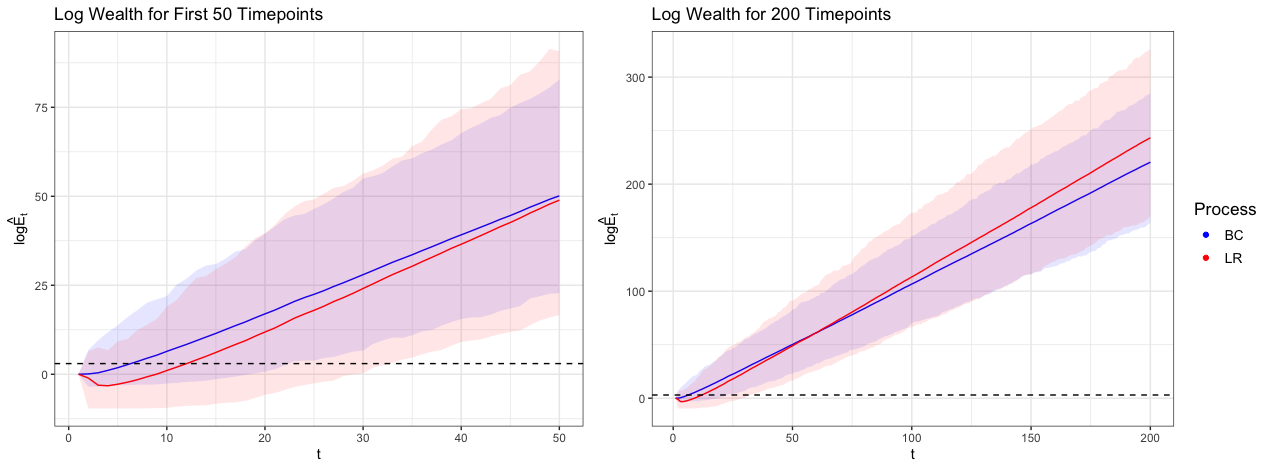}
    \caption{The averages and 95\% confidence bands for the Besag-Clifford (BC) and likelihood ratio (LR) e-processes for testing $\{\N(0,1)\}$ against $\{\N(\mu, \sigma^2):\mu \in \mathbb R, \sigma^2 > 0\}$. The results are taken from $1,000$ independent simulations. The left shows trajectories from the first $50$ time points, whereas the right shows the whole $200$ time points. Both processes use the GRAPA betting strategy. The dotted line is $\log(1/0.05)$.}
    \label{fig:bc-vs-lr}
\end{figure}

We generate $1000$ independent experiments, and sample from the alternative distribution $\N(1,4)$ for $n=200$ time points.
The LR e-process tends to dip at early time points, before recovering and increasing; in contrast, the Besag-Clifford e-process remains stable, evolving at a near-exponential rate on average (Figure~\ref{fig:bc-vs-lr}). In addition, at the $\alpha = 0.05$ significance level, the stopping time $\tau^{\tx{BC}} = \min_{t \geq 1} \{\widehat E_t \geq 1/\alpha\}$ tends to be smaller than the analogous stopping time $\tau^{\tx{LR}} = \min_{t \geq 1} \{D_t \geq 1/\alpha\}$. The wealth of the LR e-process catches up with that of the Besag-Clifford e-process near time $t=50$, although the width of the confidence band is narrower for the latter than for the former across all points. The faster rate of the LR process for increasing $t$ is due to convergence of $D(x_t \mid \bar x_{t-1})$ to $ D^*(x_t) = \N(x_t;1,4)/\N(x_t;0,1)$, resulting in an e-process with an optimal growth rate. These results suggest an intriguing strategy of running the Besag-Clifford e-process initially, then switching to the LR e-process after a certain time point, e.g., $\tau^{\tx{BC}}$, a concept known as \textit{optional continuation} \citep{grunwald2024safe}. This strategy is appealing in more general testing problems where samples may be expensive.

\section{Shapley Supercluster velocities} \label{section:application}
The Shapley Supercluster is a vast concentration of galaxies located in the Centaurus constellation, its center nearly $650$ million light-years from earth \citep{proust2006Shapley}. The Shapley galaxy dataset \citep{drinkwater2004galaxy} measures the right ascension, declination, magnitude, velocity, and velocity standard error for $n=4,215$ galaxies belonging to the Shapley supercluster. The supercluster contains several galaxy clusters, and galaxy clusters are often reflected as multiple modes in the velocity measurement \citep{roeder1990galaxy}. In this section, we illustrate our methodology with a test for the number of t-experts, or $W$, in the PoE model for the galaxy velocities in the Shapley Supercluster; see Section~\ref{section:examples}.

We split the data randomly into two equal-sized parts, $\mathcal D_a$ and $\mathcal D_b$. We compute estimates of the centers $\hat \psi_w$, the scales $\hat \sigma_w$, and the degrees of freedom $\hat \theta_w$ via score matching with stochastic gradient descent in PyTorch \citep{paszke2019pytorch} using $\mathcal D_a$. For a single observation $i \in \mathcal D_b$, null hypothesis $W_0$, and alternative hypothesis $W_1$, we set the test statistic to be
\begin{equation*}
    T(x_i) = \prod_{w_o=1}^{W_0} \lb 1 + \frac{1}{\hat \theta_{w_0}} \left( \frac{x_i - \hat \psi_{w_0}}{\hat \sigma_{w_0}} \right)^2 \rb^{\frac{\hat \theta_{w_0} + 1}{2}} \times \prod_{w_1=1}^{W_1} \lb 1 + \frac{1}{\hat \theta_{w_1}} \left( \frac{x_i - \hat \psi_{w_1}}{\hat \sigma_{w_1}} \right)^2 \rb^{\frac{-\hat \theta_{w_1} + 1}{2}}.
\end{equation*} 
Thus, $T(x_i)$ is a ULR. The number of experts are set to $W_0=5$ and $W_1=25$, meaning that the alternative is a richer model than the null. Using the implementation of HMC in blackJAX \citep{cabezas2024blackjax}, we compute the ULR process $\widehat E_t = \prod_{i=1}^t \widehat U_i(X_i)$ and fix $J=500$ and $M = 1000$. The velocity ULR process exhibits exponential decay with increasing $t$, decisively favoring the simpler model, $W=5$ (Figure~\ref{fig:Shapley}); each individual e-value is $\widehat U_i(x_i) < 1$, thus we fail to reject $W=5$ at every time point.

\begin{figure}
    \centering
    \includegraphics[width=0.95\linewidth]{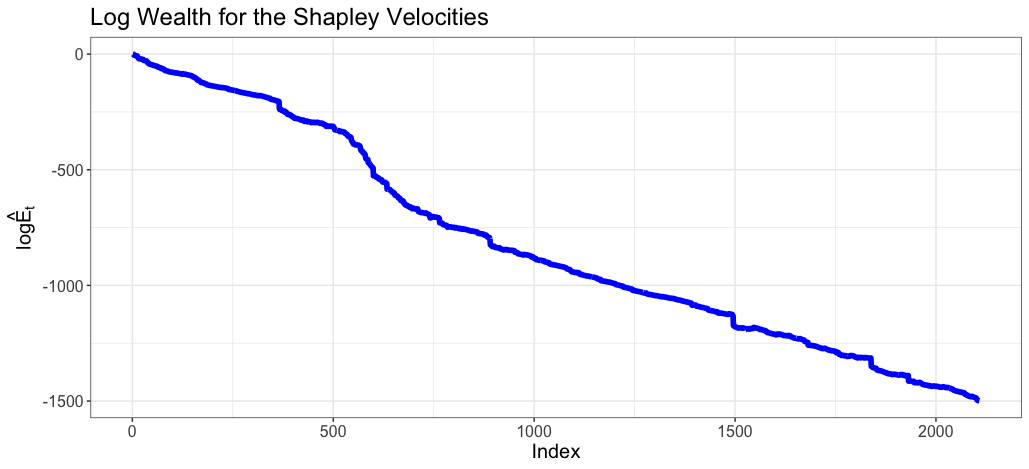}
    \caption{Log-wealth of the ULR process for group $\mathcal D_b$ of the Shapley dataset when comparing $W=5$ and $W=25$ for PoE models pretrained on $\mathcal D_a$. }
    \label{fig:Shapley}
\end{figure}

\section{Discussion} \label{section:discussion}

This article proposes Besag-Clifford e-values, in which drawing null exchangeable samples normalizes a test statistic to have expectation at most $1$ under the null. A key focus of this work is on likelihood ratio testing for unnormalized distributions, and we show that approximate log-optimality is attained for infinitely many samples and a well-mixing MCMC algorithm. Incorporating multiple chains improves the e-power in comparison to the single chain case, and we expand our methodology to encompass composite hypotheses and uncertainty quantification. The scope of our approach exceeds unnormalized likelihood ratio testing; any test statistic and hyperparameters $M>0$ and $J>0$ provide type-1 error control, so long as Algorithm~\ref{alg:besag-clifford-parallel} can be implemented. We also introduce the Besag-Clifford e-process for sequential testing, a special case being the ULR process. As $M \to \infty$, the ULR process converges to the corresponding likelihood ratio process, which has an optimal growth rate. We illustrate our theoretical findings about the AR$(1)$ process using synthetic data. We then empirically evaluate the impact of mixing, $M$, $J$, and the number of chains $S$ on the behavior of Besag-Clifford e-values under the alternative. The third simulation study compares a Besag-Clifford e-process to a universal inference e-process, where the stopping time of the former tends to be smaller than that of the latter. Finally, we provide an illustration of our approach to galaxy velocities in the Shapley Supercluster.

A natural extension of our work is multiple testing with Besag-Clifford e-values. So long as each hypothesis is of finite cardinality and consists of distributions that can be sampled from using MCMC, we can readily apply our approach to multiple testing using the e-BH procedure \citep{wang2022false}. A more difficult problem is to account for null hypotheses $\mathcal P$ of infinite cardinality. Running Algorithm~\ref{alg:besag-clifford-parallel} for infinitely many distributions is infeasible, so one would likely need to rely on alternate methods for designing valid e-values. For instance, co-sufficient (or approximately co-sufficient) sampling simulates $Y^{(1)}, \dots, Y^{(M)}$ that are (approximately) exchangeable with $X$ for any $P \in \mathcal P$; see \cite{barber2022testing} and references therein. Alternately, if $\EP[T(X)] < \infty$, one could design an approximately valid e-value by sampling $Y^{(m)} \iid P^* = \tx{arg max}_{P \in \mathcal P} \EP[T(X)]$, though this sampling strategy is not straightforward. Another possible approach is to use Monte Carlo maximum likelihood \citep{geyer1992constrained} to define an approximately valid e-value. Finally, an exciting avenue of future research is designing sampling methods for estimating the numeraire e-variable for composite nulls \citep{larsson2025numeraire}.

This article primarily focuses on the parallel method of \cite{besag1989generalized}, as it is computationally and conceptually appealing for our applications of interest. However, \cite{besag1989generalized} also introduced the \textit{serial method} for Monte Carlo significance testing. Proposition 3.7 of \cite{howes2026markov} suggests that augmenting the serial method with random permutations results in $\Ehat_M(X)$ converging to the log-optimal e-variable as $M \to \infty$ under regularity conditions for the MCMC algorithm. Hence, it would be interesting to compare the empirical e-power and computational efficiency of $\Ehat_M(X)$ when using either the parallel or serial method. More generally, it would be insightful to analyze the type-1 error control of $\Ehat_M(X)$ when $(X,Y^{(1)}, \dots, Y^{(M)})$ are not exchangeable,
along the lines of similar analyses for conformal prediction \citep{barber2023confomral}. The ergodic theorem would imply type-1 error control as $M \to \infty$ (i.e., the denominator in equation \eqref{eq:biased-theorem-1} would be equal to $1$), but validity of $\Ehat_M(X)$ for any finite $M$ may not hold without careful specification of the sampling algorithm.

\section*{Supplemental material}
The supplemental material contains proofs for results that were not proved in the text and technical details on the simulation studies and Shapley Supercluster application. Code and documentation is available at \url{github.com/adombowsky/BCevalues}.

\section*{Acknowledgments}
AD and BEE were funded in part by grants from the Parker Institute for Cancer Immunology (PICI), the Chan-Zuckerberg Institute (CZI), the Biswas Family Foundation, NIH NHGRI R01 HG012967, and NIH NHGRI R01 HG013736. AR was funded by NSF grant DMS-2310718. BEE is a CIFAR Fellow in the Multiscale Human Program.

\section*{Competing interests}
BEE is on the Scientific Advisory Board for ArrePath Inc, GSK AI for Cancer, and Freenome; she consults for Neumora.

\bibliography{sources}
\bibliographystyle{apalike}

\appendix

\section*{Supplemental material}

\section{Proofs} \label{supp:proofs}

\subsection{Proof of Theorem~\ref{thm:ergodic-log-optimal-iid}}

\begin{proof}

    First, we introduce notation related to product measures. For sets $A_0, A_1, \dots, A_M$, let $\bigtimes_{m=0}^\infty A_m$ be their Cartesian product; for $\sigma$-algebras $\mathcal F_0, \mathcal F_1, \dots, \mathcal F_M$, let $\bigotimes_{m=0}^M \mathcal F_m = \sigma\left( \lb \bigtimes_{m=0}^M A_m: A_m \in \mathcal F_m \: \forall m=0, \dots, M \rb \right)$; and for probability measures $\P_1, \dots, \P_M$, let $\bigtimes_{m=0}^M \P_m$ be their product measure. Then, let $\Omega^\infty = \bigtimes_{m=0}^\infty \Omega$, $\mathcal B^\infty = \bigotimes_{m=0}^\infty \mathcal B$, and $\mathbb L = W \times \bigtimes_{m=1}^\infty P$, where $W \in \{ P, Q \}$, defining a probability space $(\Omega^\infty, \mathcal B^\infty,\mathbb L)$. 
    
    Note that $\mathbb L = W \times \bigtimes_{m=1}^\infty P$ is the product measure for $(X, Y^{(1)}, Y^{(2)}, \dots)$ over $(\Omega^\infty, \mathcal B^\infty)$. For any $\bs \omega = (\omega_0, \omega_1, \omega_2, \dots ) \in \Omega^\infty$, let $T^{0}(\bs \omega) = T(X(\omega_0))$, $\tilde T^{(m)}(\bs \omega) = T(Y^{(m)}(\omega_m))$, and $E(\bs \omega) := E(X(\omega_0)) = q(X(\omega_0))/p(X(\omega_0))$. If $T^0(\bs \omega) = 0$, then $E(\bs \omega) = 0$ and, under the convention that $0/0=0$, $\hat E_M(\bs \omega) = 0$ for all $M \in \mathbb N$, verifying the claim. For the remainder of the proof, we assume that $\omega_0 \in \{ \omega: D(\omega) > 0 \}$ and that $\bs \omega$ is not contained in a null set for $\mathbb L$.

    Let $0<K<\infty$ be defined by $\int T(x)p(x) \diff x = 1/K$, implying $E(X) = K T(X)$. By the strong law of large numbers,
        \begin{equation*}
        \lim_{M \to \infty} \frac{1}{M}\sum_{m=1}^M \tilde T^{(m)}(\bs \omega) = 1/K     \end{equation*}
    For any $\bs \omega$ satisfying the above limit, there exists $B(\bs \omega) \geq 0$ and $M(\bs \omega) \in \mathbb N$ such that
    \begin{equation*}
        \frac{1}{M}\sum_{m=1}^M  \tilde T^{(m)}(\bs \omega) \leq B(\bs \omega) \tx{ for all } M\geq M(\bs \omega),
    \end{equation*}
    Thus, given $\epsilon>0$, since $D(\omega)< \infty$ for all $\omega \in \Omega$, there exists $M^\prime(\bs \omega) \in \mathbb N$ so that for all $M \geq M^\prime(\bs \omega)$,
    \begin{equation*}
        \frac{T^0(\bs \omega)}{M+1} + \frac{1}{M}\sum_{m=1}^M \tilde T^{(m)}(\bs \omega) \left( 1 - \frac{M}{M+1} \right) < \epsilon,
    \end{equation*}
    meaning that $\lim_{M \to \infty} \{ T^0(\bs \omega)/(M+1) + \sum_{m=1}^M \tilde T^{(m)}(\bs \omega)/(M+1)\} = 1/K$. Since $T^0(\bs \omega) = T(\omega_0) > 0$, continuity implies that for large enough $M$,
    \begin{equation*}
        \bigg |\frac{(M+1)}{T^0(\bs \omega) + \sum_{m=1}^M \tilde T^{(m)}(\bs \omega)} - K  \bigg | < \frac{\epsilon}{T^0(\bs \omega)}. \qedhere
    \end{equation*}
\end{proof}

\subsection{Proof of Theorem~\ref{thm:Besag-parallel}}

\begin{proof}
     Initially, suppose that $T(x) = q(x)/p(x)$. First, note that $Y^{(m)} \ind X \mid Y^{(0)}$ and, for any $y^{(0)}$, the weak law of large numbers guarantees that,
    \begin{equation*}
        \frac{1}{M} \sum_{m=1}^M T(Y^{(m)}) \overset{pr.}{\longrightarrow} \Delta^J(y^{(0)}) \mid Y^{(0)} = y^{(0)}.
    \end{equation*}
    Then $(1/M)\sum_{m=1}^M T(Y^{(m)}) \overset{pr.}{\to} \Delta^J(Y^{(0)})$ by the dominated convergence theorem under $V^J(x,\cdot)$, the conditional distribution of $Y^{(0)} \mid X=x$. The continuous mapping theorem yields that for all $x$,
    \begin{equation} \label{eq:conditional-pr-limit}
        \frac{(M+1)T(x)}{ T(x) + \sum_{m=1}^M T(Y^{(m)}) } \overset{pr.}{\longrightarrow} \frac{T(x)}{\Delta^J(Y^{(0)})}, %\bigg | X=x. 
    \end{equation}
    with respect to $V^J(x,\cdot)$. Applying the dominated convergence theorem with $X \sim W$ gives the result. For $T(x) \propto q(x)/p(x)$, we can multiply the numerator and denominator of $\Ehat_M(x)$ by the unknown normalizing constant $K$, yielding the same convergence in probability statement as in equation \eqref{eq:conditional-pr-limit}. Therefore, the result holds for any $T(x) \propto q(x)/p(x)$.
\end{proof}

\subsection{Proof of Proposition~\ref{prop:mixing-time}}
\begin{proof}
    \begin{gather*}
        \KL(V^{J_{\KL}(\epsilon)}(y,\cdot),P) < 2 \epsilon^2 \tx{ for all } y \\
        \implies  q(y)\KL(V^{J_{\KL}(\epsilon)}(y,\cdot),P) < 2 \epsilon^2 q(y) \tx{ for all } y  \\
        \implies \int q(y)\KL(V^{J_{\KL}(\epsilon)}(y,\cdot),P) \diff y < 2 \epsilon^2.
    \end{gather*}
    The result follows by using the inequality in equation \eqref{eq:delta-e-power}. 
\end{proof}

Using a similar technique, we can bound the Lesbesgue measure of the difference between the numerator and denominator in $\Delta^J(Y^{(0)})$.
\begin{prop} \label{prop:delta-difference}
    Let $J(\epsilon)$ be the mixing time of $\Vdot$. Then,
    \begin{equation}
       \frac{1}{2}\int  \int q(y) |v^{J(\epsilon)}(y,y^{(0)}) - p(y^{(0)}) | \diff y \diff y^{(0)} < \epsilon.
    \end{equation}
\end{prop}
\begin{proof}
    By definition of the Markov chain mixing time, 
    \begin{gather*}
        \frac{1}{2} \int | v^{J(\epsilon)}(y,y^{(0)}) - p(y^{(0)})| \diff y^{(0)} < \epsilon \tx{ for all } y \\
        \implies \frac{1}{2} \int q(y)| v^{J(\epsilon)}(y,y^{(0)}) - p(y^{(0)})| \diff y^{(0)} < \epsilon q(y) \tx{ for all }y \\ 
        \implies \frac{1}{2} \int \int q(y)| v^{J(\epsilon)}(y,y^{(0)}) - p(y^{(0)})| \diff y^{(0)} \diff y < \epsilon \qedhere.
    \end{gather*}
\end{proof}
Without making an additional assumptions (e.g., that there is a non-zero lower bound on $v^J(y,\cdot)$), we cannot translate this result directly to the e-power. However, Proposition~\ref{prop:delta-difference} does provide insight onto the total variation distance between $P$ and $Q^*$, i.e. the distributions of $Y^{(0)}$ under the null and alternative, respectively.

\subsection{Proof of Proposition~\ref{prop:Besag-multi-chain}}
\begin{proof}
    
    An application of the weak law gives that, conditional on $X=x$,
    \begin{equation*}
        \frac{1}{S} \sum_{s=1}^S \frac{E(x)}{\Delta^J(Y_s^{(0)})} \overset{pr.}{\longrightarrow} E(x)\EQ \bigg [ \frac{1}{\Delta^J(Y^{(0)})} \mid X=x \bigg],
     \end{equation*}
     and applying the dominated convergence theorem completes the proof.
\end{proof}

\subsection{Proof of Theorem~\ref{thm:null-exchanged-e-power}}

\begin{proof}
        First, we have that
        \begin{equation*}
            \EQ[\log \Ehat_M(X)] = \EQ[\log T(X)] + \E^\Q \lb - \log \left( \frac{1}{M+1} T(X) + \frac{1}{M+1} \sum_{m=1}^M \tilde T^{(m)} \right) \rb.
        \end{equation*}
        By Jensen's inequality,
        \begin{gather*} \label{eq:Jensen's-power}
            \EQ \lb - \log \left( \frac{1}{M+1} T(X) + \frac{1}{M+1} \sum_{m=1}^M \tilde T^{(m)} \right) \rb \\ > -\log \left( \frac{1}{M+1}\sum_{m=1}^M \E[\tilde T^{(m)}] + \frac{\EQ[T(X)]}{M+1} \right) = -\log \left( \frac{M \E[\tilde T]}{M+1} + \frac{\EQ[T(X)]}{M+1} \right), 
        \end{gather*}
        where the inequality is strict since $-\log t$ is strictly convex and $T(X), \tilde T^{(1)}, \dots, \tilde T^{(M)}$ are non-degenerate when $Q$ is true.  
    \end{proof}

\subsection{Proof of Corollary~\ref{corr:finite-M}}

\begin{proof}
   The right hand side of equation \eqref{eq:null-exchanged-e-power} is positive if and only if $\EQ[T(X)] - 1 < M(1 - \E[\tilde T])$. The existence of $M_*^Q$ then follows under the assumption $\E[\tilde T] = \E^P[T(X)] < 1$. If $\sup_{Q \in \mathcal Q} \EQ[T(X)] < \infty$, the right hand side of equation \eqref{eq:null-exchanged-e-power} is positive for all $Q \in \mathcal Q$ at $M_* = \min_{M \in \mathbb N} \{ \sup_{Q \in \mathcal Q} [T(X)] < M(1 - \E[\tilde T]) + 1 \}$, proving existence of $M_*.$ Lastly, if $\EQ[T(X)] \leq 1$, then the remainder term on the right hand side of equation \eqref{eq:null-exchanged-e-power} is always nonnegative, so $\EQ[\log \Ehat_M(X)] > \EQ[\log T(X)]$ for any $M.$
\end{proof}

\section{Details on simulation studies}
R code is available at \url{github.com/adombowsky/BCevalues}. The AR$(1)$ process (Section~\ref{sect:ar1-sims}) is simulated via equation \eqref{eq:AR1-model} in R. The PoE model (Section~\ref{sect:sequential-product-of-experts}) is sampled from using the \texttt{MCMCpack} \citep{martin2011mcmcpack} and the \texttt{MCMCmetrop1R()} function. The different runs use distinct random number seeds for sampling. All hyperparameters in \texttt{MCMCmetrop1R()} are set to their defaults. Confidence bands (Figure~\ref{fig:ulr-process-t-product} and Figure~\ref{fig:bc-vs-lr}) refer to the $2.5\%$ and $97.5\%$ empirical quantiles across the replications at each time point $t$. For the sequential test of $\mathcal P = \{\N(0,1)\}$ vs. $\mathcal Q = \{ \N(\mu, \sigma^2):\mu \in \mathbb R, \sigma^2 > 0\}$ (Section~\ref{sect:sequential-composite-alt}), the GRAPA betting strategy is implemented via the definition of an empirically adaptive e-process \citep{ramdas2025hypothesis}, which is
\begin{equation} \label{eq:supplement-GRAPA}
    \lambda_t \in \underset{\lambda \in [0,1]}{\tx{arg max}} \frac{1}{t-1} \sum_{i=1}^{t-1} \log (1 - \lambda + \lambda \widehat U_i(X_i) ),
\end{equation}
and is calculated at every time point $t=1, \dots, 50$ using the \texttt{optimize()} function. GRAPA for the LR process is computed similarly with $\widehat U_i(X_i)$ in equation \eqref{eq:supplement-GRAPA} replaced with $D(X_i \mid X_{1:(i-1)})$.

\section{Details on application}

The Shapley galaxy dataset is \href{https://sites.psu.edu/astrostatistics/datasets-shapley-galaxy-dataset/}{publicly available} as supplementary material for \cite{feigelson2012modern}. Python code for our application to the Shapley dataset is available at \url{github.com/adombowsky/BCevalues}. We center and scale the velocity to have zero mean and unit standard deviation. Score matching is implemented with the Adam stochastic gradient descent algorithm in \texttt{torch.optim.Adam} \citep{adam2014method}. The unnormalized log-likelihood for the PoE model is $\mathcal L(x_{\mathcal D_a}) = \sum_{i \in \mathcal D_a} \mathcal L(x_i)$, where
\begin{equation*}
    \mathcal L(x_i) = - \sum_{w_0=1}^{W_0} \left( \frac{\theta_{w_0}+1}{2} \right) \log \left( 1 + \frac{1}{\theta_{w_0}} \left( \frac{x_i - \psi_{w_0}}{\sigma_{w_0}} \right)^2  \right).
\end{equation*}
The score-matching loss is $s(x_{\mathcal D_a}) = (1/|\mathcal D_a|) \sum_{i= \in \mathcal D_a} s(x_i)$, with
\begin{equation*}
    s(x_i) =  \frac{1}{2}\left( \frac{d}{dx} \mathcal L(x_i) \right)^2 + \frac{d^2}{dx^2} \mathcal L(x_i).
\end{equation*}

\begin{figure}[t]
    \centering
    \includegraphics[width=0.9\linewidth]{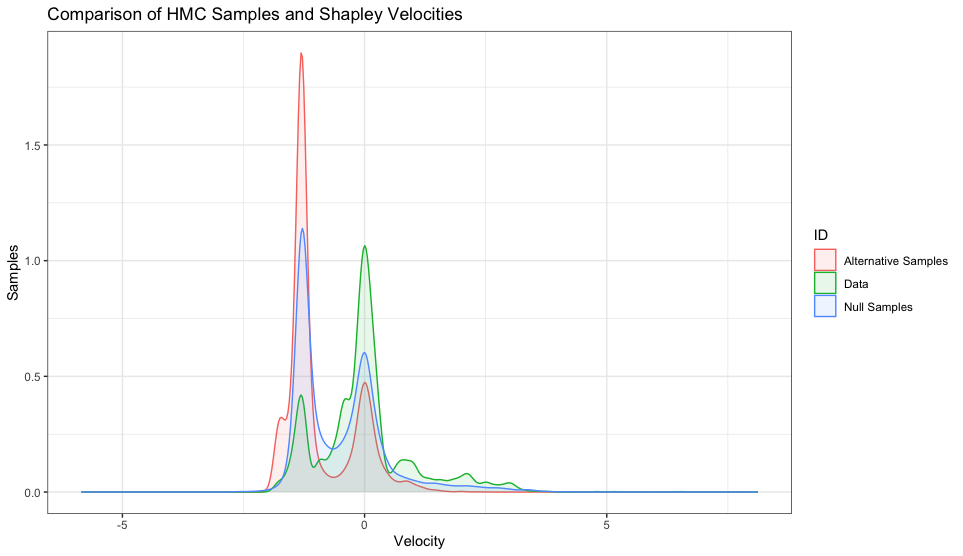}
    \caption{Density estimates of the scaled Shapley velocities (both $\mathcal D_a$ and $\mathcal D_b$) and $40000$ HMC samples from the null pre-trained null model ($W=5$) and pre-trained alternative model ($W=25$).}
    \label{fig:shapley-comparison}
\end{figure}

For both hypotheses, the PoE model is trained using $3000$ iterations of Adam. The HMC samples are generated with the HMC sampler in \texttt{blackjax.hmc} \citep{cabezas2024blackjax}. We set the step size parameter to be $0.1$, the inverse mass matrix as the identity, and the number of integration steps to be $1000$. In our code, we compute the statistic $\tilde T(x) = \log T(x)$, then calculate $\widehat U_i(X_i)$ with
\begin{equation*} \label{eq:log-sum-exp}
\begin{gathered}
        \log \frac{\widehat U_i(X)}{M+1} = \tilde T(X) - \tilde T^* - \log\left( \exp(\tilde T(X)-\tilde T^*) + \sum_{m=1}^M\exp (\tilde T(Y^{(m)}) - \tilde T^*)\right), \\
        \tilde T^* = \max \lb \tilde T(X), \tilde T(Y^{(1)}), \dots, \tilde T(Y^{(M)}) \rb.
\end{gathered}
\end{equation*}
As a follow-up to the conclusions of the ULR e-process, we simulate $10,000$ iterations across $4$ independent chains from both the null and the alternative model using \texttt{blackjax.hmc}. Samples from the pre-trained models tend to overestimate and underestimate the lower and higher mode in the velocities (Figure~\ref{fig:shapley-comparison}). However, the modes of the null samples are closer to those in the actual data than the alternative model. Since the higher mode has a larger null density estimate, this may lead to the log-likelihood of $\mathcal D_b$ being higher for the null than the alternative, ultimately resulting in decay of the ULR e-process.

\end{document}